\documentclass[11pt]{article}
\usepackage{geometry}
\usepackage[utf8]{inputenc}
\usepackage{stringstrings}

\usepackage{fullpage}

\usepackage{wrapfig}
\usepackage{amsmath}
\usepackage{amsfonts}
\usepackage{amssymb}
\usepackage{amsthm}
\usepackage{bm}
\usepackage{bbm}
\usepackage{dsfont}

\usepackage[ruled,noresetcount]{algorithm2e}

\usepackage{hyperref}
\usepackage{cleveref}

\usepackage{mathtools}
\usepackage{multirow}

\usepackage{environ}

\usepackage{ifthen}

\newtheorem{theorem}{Theorem}

\newtheorem{lemma}[theorem]{Lemma}

\theoremstyle{definition}
\newtheorem{definition}[theorem]{Definition}

\DeclareMathOperator*{\E}{E}

\newcommand{\R}{\mathbb R}

\newcommand{\eps}{\varepsilon}

\newcommand{\abs}[1]{\left\lvert {#1} \right\rvert}

\newcommand{\norm}[2]{\left\lVert {#1} \right\rVert_{{#2}}}
\newcommand{\pnorm}[2]{\left\lVert {#1} \right\rVert_{{#2}}}

\newcommand{\ep}[2][]{\E_{{#1}}\!\left[ {#2} \right]}
\newcommand{\epcond}[3][]{\E_{{#1}}\!\left[ {#2} \, \middle| \, {#3} \right]}
\newcommand{\prb}[2][]{\Pr_{#1}\!\left[ #2 \right]}

\newcommand{\indicator}[1]{\left[ {#1} \right]}

\newcommand{\set}[1]{\left\{ {#1} \right\}}
\newcommand{\setbuilder}[2]{\left\{ {#1} \, \middle| \, {#2} \right\}}

\newcommand{\floor}[1]{\left\lfloor {#1} \right\rfloor}

\newcommand{\xor}{\oplus}
\newcommand{\bigxor}{\bigoplus}

\DeclarePairedDelimiterX{\infdivx}[2]{(}{)}{{#1} \,\delimsize\|\, {#2}}

\NewEnviron{splitEquation}{\begin{equation}\begin{split} \BODY \end{split}\end{equation}}
\NewEnviron{splitEquation*}{\begin{equation*}\begin{split} \BODY \end{split}\end{equation*}}

\makeatletter
\providecommand*{\cupdot}{%
  \mathbin{%
    \mathpalette\@cupdot{}%
  }%
}
\newcommand*{\@cupdot}[2]{%
  \ooalign{%
    $\m@th#1\cup$\cr
    \hidewidth$\m@th#1\cdot$\hidewidth
  }%
}
\providecommand*{\bigcupdot}{%
  \mathop{%
    \mathpalette\@bigcupdot{}%
  }%
}
\newcommand*{\@bigcupdot}[2]{%
  \ooalign{%
    $\m@th#1\bigcup$\cr
    \hidewidth$\m@th#1\cdot$\hidewidth
  }%
}
\makeatother

\makeatletter
\newcommand{\bigmodeproduct}[1]{
  \mathop{
    \mathchoice{\vcenter{\hbox{{\huge $\m@th\mkern-2mu\times\mkern-2mu$}$_{#1}$}}}          
               {\vcenter{\hbox{{\LARGE $\m@th\mkern-2mu\times\mkern-2mu$}$_{#1}$}}}         
               {\vcenter{\hbox{{$\m@th\mkern-2mu\times\mkern-2mu$}$_{#1}$}}}                
               {\vcenter{\hbox{{\footnotesize $\m@th\mkern-2mu\times\mkern-2mu$}$_{#1}$}}}  
  }\displaylimits
}
\makeatother

\newcommand{\restateable}[4][theorem]{
  \ifthenelse{\equal{#1}{theorem}}{
    \begin{theorem}\label{thm:#2}
      {#3}
    \end{theorem}

    \convertchar[e]{#2}{-}{ }
    \capitalizewords[e]{\thestring}
    \noblanks[e]{\thestring}
    \retokenize[q]{\thestring}

    \expandafter\newcommand\csname restateThm\thestring\endcsname{
      {
        \def\thetheorem{\ref{thm:#2}}
        \begin{theorem}
          {#4}
        \end{theorem}
        \addtocounter{theorem}{-1}
      }
    }%
  }{\ifthenelse{\equal{#1}{lemma}}{
    \begin{lemma}\label{lem:#2}
      {#3}
    \end{lemma}

    \convertchar[e]{#2}{-}{ }
    \capitalizewords[e]{\thestring}
    \noblanks[e]{\thestring}
    \retokenize[q]{\thestring}

    \expandafter\newcommand\csname restateLem\thestring\endcsname{
      {
        \def\thelemma{\ref{lem:#2}}
        \begin{lemma}
          {#4}
        \end{lemma}
        \addtocounter{lemma}{-1}
      }
    }%
  }{}}
}

\title{
A Sparse Johnson-Lindenstrauss Transform using Fast Hashing
}

\author{
        Jakob Bæk Tejs Houen\footnote{Research supported by Investigator Grant 16582, Basic Algorithms Research Copenhagen (BARC), from the VILLUM Foundation.}\ \\ \small  BARC, University of Copenhagen \\ \small \texttt{jakob@tejs.dk}
    \and
        Mikkel Thorup$^*$ \\ \small BARC, University of Copenhagen \\ \small \texttt{mikkel2thorup@gmail.com}
}

\date{}

\begin{document}

\setcounter{page}{0}
\maketitle

\begin{abstract}
    The \emph{Sparse Johnson-Lindenstrauss Transform} of Kane and Nelson (SODA 2012) provides a linear dimensionality-reducing map $A \in \R^{m \times u}$ in $\ell_2$ that preserves distances up to distortion of $1 + \eps$ with probability $1 - \delta$, where $m = O(\eps^{-2} \log 1/\delta)$ and each column of $A$ has $O(\eps m)$ non-zero entries.
    The previous analyses of the Sparse Johnson-Lindenstrauss Transform all assumed access to a $\Omega(\log 1/\delta)$-wise independent hash function.
    The main contribution of this paper is a more general analysis of the Sparse Johnson-Lindenstrauss Transform with less assumptions on the hash function.
    We also show that the \emph{Mixed Tabulation hash function} of Dahlgaard, Knudsen, Rotenberg, and Thorup (FOCS 2015) satisfies the conditions of our analysis, thus giving us the first analysis of a Sparse Johnson-Lindenstrauss Transform that works with a practical hash function.
\end{abstract}

\thispagestyle{empty}

\newpage

\section{Introduction}

Dimensionality reduction is an often applied technique to obtain a speedup when working with high dimensional data.
The basic idea is to map a set of points $X \subseteq \R^u$ to a lower dimension while approximately preserving the geometry.
The Johnson-Lindenstrauss lemma~\cite{JL1984} is a foundational result in that regard.

\begin{lemma}[\cite{JL1984}]
    For any $0 < \eps < 1$, integers $n, u$, and $X \subseteq \R^u$ with $\abs{X} = n$, there exists a map $f \colon X  \to \R^m$ with $m = O(\eps^{-2} \log n)$ such that
    \[
        \forall w, w' \in X, \; \abs{\norm{f(w) - f(w')}{2} - \norm{w - w'}{2}} \le \eps \norm{w - w'}{2}  .
    \]
\end{lemma}

It has been shown in \cite{Alon17,Larsen17} that the target dimension $m$ is optimal for nearly the entire range of $n, u, \eps$. More precisely, for any $n, u, \eps$ there exists a set of points $X \subseteq \R^u$ with $|X| = n$ such that for any map $f \colon X \to \R^m$ where the Euclidean norm is distorted by at most $(1 \pm \eps)$ must have $m = \Omega(\min\set{u, n, \eps^{-2} \log(\eps^2 n)})$.

All known proofs of the Johnson-Lindenstrauss lemma constructs a linear map $f$.
The original proof of Johnson and Lindenstrauss~\cite{JL1984} chose $f(x) = \Pi x$ where $\Pi \in \R^{m \times u}$ is an appropriately scaled orthogonal projection into a random $m$-dimensional subspace.
Another simple construction is to set $f(x) = \frac{1}{\sqrt{m}} A x$ where $A \in \R^{m \times u}$ and each entry is an independent Rademacher variable.\footnote{A Rademacher variables, $X$, is a random variable that is chosen uniformly in $\pm 1$, i.e., $\prb{X = 1} = \prb{X = -1} = \tfrac{1}{2}$.}
In both cases, it can be shown that as long as $m = \Omega(\eps^{-2} \log 1/\delta)$ then
\begin{align}\label{eq:distributional-jl}
    \forall w \in \R^u, \; \prb{\abs{\norm{f(w)}{2}^2 - \norm{w}{2}^2} \ge \eps \norm{w}{2}^2} \le \delta .
\end{align}
The Johnson-Lindenstrauss lemma follows by setting $\delta < 1/\binom{n}{2}$ and taking  $w = z - z'$ for all pairs $z, z' \in X$ together with a union bound.
\eqref{eq:distributional-jl} is also known as the distributional Johnson-Lindenstrauss lemma and it has been shown that the target dimension $m$ is tight, more precisely, $m$ must be at least $\Omega(\min\set{u, \eps^{-2}\log 1/\delta})$~\cite{Jayram13,Kane11}.

\paragraph*{Sparse Johnson-Lindenstrauss Transform.}
One way to speed up the embedding time is replacing the dense $A$ of the above construction by a sparse matrix.
The first progress in that regard came by Achlioptas in~\cite{Achlioptas01} who showed that $A$ can be chosen with i.i.d. entries where $A_{i j} = 0$ with probability $2/3$ and otherwise $A_{i j}$ is chosen uniformly in $\pm \sqrt{\frac{3}{m}}$.
He showed that this construction can achieve the same $m$ as the best analyses of the Johnson-Lindenstrauss lemma.
Hence this achieves essentially a 3x speedup, but the asymptotic embedding time is still $O(m \norm{x}{0})$ where $\norm{x}{0}$ is number of non-zeros of $x$.

Motivated by improving the asymptotic embedding time, Kane and Nelson in~\cite{Kane14}, following the work in~\cite{Dasgupta10,Kane10,Braverman10}, introduced the Sparse Johnson-Lindenstrauss Transform which maps down to essentially optimal dimension $m = O(\eps^{-2} \log n)$ and only has $s = O(\eps^{-1} \log n)$ non-zeros entries per column.
This speeds up the embedding time to $O(\eps^{-1} \log n \norm{x}{0}) = O(\eps m \norm{x}{0})$ thus improving the embedding time by a factor of $\eps^{-1}$.
It nearly matches a sparsity lower bound by Nelson and Nguyen~\cite{Nelson13} who showed that any sparse matrix needs at least $s = \Omega(\eps^{-1} \log( n) / \log (1/\eps))$ non-zeros per column.

\paragraph*{Using Hashing.}
When the input dimension, $u$, is large it is not feasible to store the matrix $A$ explicitly.
Instead, we use a hash function to calculate the non-zero entries of $A$.
Unfortunately, the previous analyses of the Sparse Johnson-Lindenstrauss Transform~\cite{Kane14,Cohen18} assume access to a $\Omega(\log 1/\delta)$-wise independent hash function which is inefficient.
This motivates the natural question:
\begin{center}
    \emph{
        What are the sufficient properties we need of the hash function for a Sparse Johnson-Lindenstrauss Transform to work?
    }
\end{center}
The goal of this work is to make progress on this question.
In particular, we provide a new analysis of a Sparse Johnson-Lindenstrauss Transform with fewer assumptions on the hash function.
This improved analysis allows us to conclude that there exists a Sparse Johnson-Lindenstrauss Transform that uses Mixed Tabulation hashing which is efficient.

\paragraph*{Mixed Tabulation Hashing.}
Before introducing Mixed Tabulation hashing, we will first discuss \emph{Simple Tabulation hashing} which was introduced by Zobrist~\cite{Zobrist70}.
Simple Tabulation hashing takes an integer parameter $c > 1$, and we view a key $x \in [u] = \set{0, \ldots, u - 1}$ as a vector of $c$ characters, $x_0, \ldots, x_{c - 1} \in \Sigma = [u^{1/c}]$.
For each character, we initialize a fully random table $T_i \colon \Sigma \to [2^r]$ and the hash value of $x$ is then calculated as
\[
    h(x) = T_0[x_0] \xor \ldots \xor T_{c - 1}[x_{c - 1}] ,
\]
where $\xor$ is the bitwise XOR-operation.
We say that $h$ is a Simple Tabulation hash function with $c$ characters.

We can now define \emph{Mixed Tabulation hashing} which is a variant of Simple Tabulation hashing that was introduced in~\cite{Dahlgaard15}.
As with Simple Tabulation hashing, Mixed Tabulation hashing takes $c > 1$ as a parameter, and it takes a further integer parameter $d \ge 1$.
Again, we view a key $x \in [u]$ as vector of $c$ characters, $x_0, \ldots, x_{c - 1} \in \Sigma = [u^{1/c}]$.
We then let $h_1 \colon \Sigma^c \to [2^r]$, $h_2 \colon \Sigma^c \to \Sigma^d$, and $h_3 \colon \Sigma^d \to [2^r]$ be independent Simple Tabulation hashing.
Mixed Tabulation hashing is then defined as follows
\[
    h(x) = h_1(x) \xor h_3(h_2(x)) .
\]
We say that $h$ a mixed tabulation hash function with $c$ characters and $d$ derived characters.
We call $h_2(x) \in \Sigma^d$ the \emph{derived} characters.
Mixed Tabulation hashing can be efficiently implemented by storing $h_1$ and $h_2$ as a single table with entries in $[2^r]\times \Sigma^d$, so the whole hash function can be computed with just $c+d$ lookups.

\paragraph*{Our Contributions.}
Our main contribution is a new analysis of a Sparse Johnson-Lindenstrauss Transform that does not rely on the high independence of the hash function.
Instead we show that it suffices that the hash function supports a decoupling-decomposition combined with strong concentration bounds. 

We show that Mixed Tabulation hashing satisfies these conditions.
This gives the first instance of a practical hash function that can support a Sparse Johnson-Lindenstrauss Transform.

\subsection{Sparse Johnson-Lindenstrauss Transform}\label{sec:construction}
As mentioned earlier, the Sparse Johnson-Lindenstrauss Transform was introduced by Kane and Nelson~\cite{Kane14} and they provided two different constructions with the same sparsity.
Later a simpler analysis was given in~\cite{Cohen18} which also generalized the result to a more general class of constructions.
In this paper, we will only focus on one of the constructions which is described below.

Before we discuss the construction of the Sparse Johnson-Lindenstrauss Transform, we will first consider the related CountSketch which was introduced in~\cite{Charikar03} and was analyzed for dimensionality reduction in~\cite{Thorup12}.
In CountSketch, we construct the matrix $A$ as follows: We pick a pairwise independent hash function, $h \colon [u] \to [m]$, and a 4-wise independent sign function $\sigma \colon [u] \to \set{-1, 1}$.
For each $x \in [u]$, we set $A_{h(x), x} = \sigma(x)$ and the rest of the $x$'th column to $0$.
Clearly, this construction has exactly 1 non-zero entry per column.
It was shown in~\cite{Thorup12} that if $m = \Omega(\eps^{-2} \delta^{-1})$ then it satisfies the distributional Johnson-Lindenstrauss lemma, \cref{eq:distributional-jl}.
The result follows by bounding the second moment of $\norm{A x}{2}^2 - \norm{x}{2}^2$ for any $x \in \R^d$ and then apply Chebyshev's inequality.

The bad dependence in the target dimension, $m$, on the failure probability, $\delta$, is because we only use the second moment.
So one might hope that you can improve the dependence by looking at higher moments instead.
Unfortunately, it is not possible to improve the dependence for general $x \in \R^d$, and it is only possible to improve the dependence if $\norm{x}{\infty}^2/\norm{x}{2}^2$ is small.
Precisely, how small $\norm{x}{\infty}^2/\norm{x}{2}^2$ has to be, has been shown in~\cite{Freksen18}.
So to improve the dependence on $\delta$, we need to increase the number of non-zero entries per column.

We are now ready to describe the construction of the Sparse Johnson-Lindenstrauss Transform.
The construction is to concatenate $s$ CountSketch matrices and scale the resulting matrix by $\tfrac{1}{\sqrt{s}}$.
This clearly gives a construction that has $s$ non-zero entries per column and as it has been shown in~\cite{Kane14,Cohen18} if $s = \Omega(\eps^{-1} \log(1/\delta))$ then we can obtain the optimal target dimension $m = O(\eps^{-2} \log(1/\delta))$.
More formally, we construct the matrix $A$ as follows:
\begin{enumerate}
    \item We pick a hash function, $h \colon [s] \times [u] \to [m/s]$ and a sign function $\sigma \colon [s] \times [u] \to \set{-1, 1}$.
    \item For each $x \in [u]$, we set $A_{i \cdot m/s + h(i, x), x} = \tfrac{\sigma(i, x)}{\sqrt{s}}$ for every $i \in [s]$ and the rest of the $x$'th column to $0$.
\end{enumerate}
In the previous analyses~\cite{Kane14,Cohen18}, it was shown that if $h$ and $\sigma$ are $\Omega(\log 1/\delta)$-wise independent then the construction works.
Unfortunately, it is not practical to use a $\Omega(\log 1/\delta)$-wise independent hash function so the goal of this work is to obtain an analysis of a Sparse Johnson-Lindenstrauss Transform with fewer assumptions about the hash function.
In particular, we relax the assumptions of the hash function, $h$, and the sign function, $\sigma$, to just satisfying a decoupling-decomposition and a strong concentration property.
The formal theorem is stated in \Cref{sec:proof_overview}.

We also show that Mixed Tabulation satisfies these properties and thus that the Sparse Johnson-Lindenstrauss Transform can be implemented using Mixed Tabulation.
Let us describe more formally, what we mean by saying that Mixed Tabulation can implement the Sparse Johnson-Lindenstrauss Transform.
We let $h_1 \colon \Sigma^c = [u] \to [m/s]$, $h_2 \colon \Sigma^c \to \Sigma^d$, and $h_3 \colon \Sigma^d \to [m/s]$ be the independent Simple Tabulation hash functions that implement the Mixed Tabulation hash function, $h_1(x) \xor h_3(h_2(x))$.
We then extend it to the domain $[s] \times [u]$ as follows:
\begin{enumerate}
    \item Let $h'_2 \colon [s] \times \Sigma^c \to \Sigma^d$ be defined by $h'_2(i, x) = h_2(x) \xor \underbrace{(i, \ldots, i)}_{d\text{ times}}$, i.e., each derived character gets xor'ed by $i$.
    \item We then define $h \colon [s] \times [u] \to [m/s]$ and $\sigma \colon [s] \times [u] \to \set{-1, 1}$ by $h(i, x) = h_1(x) \xor h_3(h'_2(i, x))$ and $\sigma(i, x) = \sigma_1(x) \cdot \sigma_3(h'_2(i, x))$, where $h_1$ and $h_3$ are the Simple Tabulation hash functions described above, and $\sigma_1 \colon \Sigma^c \to \set{-1, 1}$ and $\sigma_3 \colon \Sigma^d \to \set{-1, 1}$ are independent Simple Tabulation functions.
\end{enumerate}


\subsection{Hashing Speed}
When we use tabulation schemes, it is often as a fast alternative to $\Omega(\log n)$-independent hashing.
Typically, we implement a $q$-independent hash function using a degree $q-1$ polynomial in $O(q)$ time, and Siegel~\cite{siegel04hash} has proved that this is best possible unless we use large space.
More precisely, for some key domain $[u]$, if we want to do $t < q$ memory accesses, then we need space at least $u^{1/t}$.
Thus, if we want higher than constant independence but still constant evaluation times, then we do need space $u^{\Omega(1)}$.
In our application, we have to compute many hash values simultaneously, so an alternative strategy would be to evaluate the polynomial using multi-point evaluation.
This would reduce the time per hash value to $O(\log^2 q)$ but this is still super constant time.

With tabulation hashing, we use tables of size $O(|\Sigma|)$ where $|\Sigma|=u^{1/c}$ and $c=O(1)$.
The table lookups are fast if the tables fit in cache, which is easily the case for 8-bit characters.
In connection with each lookup, we do a small number of very fast AC$^0$ operations: a cast, a bit-wise xor, and a shift.
This is incomparable to polynomial in the sense of fast cache versus multiplications, but the experiments from~\cite[Table 1]{Aamand20} found Simple Tabulation hashing to be faster than evaluation a 2-wise independent polynomial hashing.

Tabulation schemes are most easily compared by the number of lookups.
Storing $h_1$ and $h_2$ in the same table, Mixed Tabulation hashing uses $c+d$ lookups.
With $d=c$, the experiments from~\cite{Aamand20} found Mixed Tabulation hashing to be slightly more than twice as slow as Simple Tabulation hashing, and the experiments from~\cite{Dahlgaard17} found Mixed Tabulation hashing to be about as fast as 3-wise independent polynomial hashing.
This motivates our claim that Mixed Tabulation hashing is practical.

In theory, we could also use a highly independent hash function that uses large space, but we don't know of any efficient construction.
Siegel states about his construction, it is ``far too slow for any practical application''~\cite{siegel04hash}, and while Thorup~\cite{Tho13:simple-simple} has presented a simpler construction than Siegel's, it is still not efficient. 
The experiments in~\cite{Aamand20} found it to be more than an order magnitude slower than Mixed Tabulation hashing.

\section{Related Work}

\paragraph*{Even Sparser Johnson-Lindenstrauss Transforms.}
As touched upon earlier, there is a lower bound by Nelson and Nguyen~\cite{Nelson13} that rules out significant improvements, but never the less there has been research into sparser embedding.
In the extreme, Feature Hashing of~\cite{Weinberger09} considers the case of $s = 1$.
The lower bound excludes Feature Hashing from working for all vectors, but in~\cite{Freksen18} they gave tight bounds for which vectors it works in terms of the measure $\norm{w}{\infty}^2/\norm{w}{2}^2$.
This was later generalized in~\cite{Jagadeesan19} to a complete understanding between the tradeoff between $s$ and the measure $\norm{w}{\infty}^2/\norm{w}{2}^2$.
In this paper, we will only focus on the case $s = \Theta(\eps^{-1} \log 1/\delta)$ and $m = \Theta(\eps^{-2} \log 1/\delta)$

\paragraph*{Fast Johnson-Lindenstrauss Transform.}
Another direction to speed-up the evaluation of Johnson-Lindenstrauss transforms is to exploit dense matrices with fast matrix-vector multiplication.
This was first done by Ailon and Chazelle~\cite{Ailon09} who introduced the Fast Johnson-Lindenstrauss Transform.
Their original construction was recently~\cite{Fandina22B} shown to give an embedding time $O(u \log u + m (\log 1/\delta + \eps \log^2(1/\delta)/\log(1/\eps)) )$.

This has generated a lot follow-up work that has tried to improve the running to a clean $O(u \log u)$.
Some of the work sacrifice the optimal target dimension, $m = O(\eps^{-2} \log 1 / \delta)$, in order to speed-up the construction, and are satisfied with sub-optimal $m = O(\eps^{-2} \log n \log^4 u)$~\cite{Krahmer11}, $m = O(\eps^{-2} \log^3 n)$~\cite{Do09}, $m = O(\eps^{-1} \log^{3/2} n \log^{3/2} u + \eps^{-2} \log n \log^4 u)$~\cite{Krahmer11}, $m = O(\eps^{-2} \log^2 n)$~\cite{Hinrichs11,Vybiral11,Freksen20}, and $m = O(\eps^{-2} \log n \log^2(\log n) \log^3 u)$~\cite{Jain22}.
Another line of progress is to assume that the target dimension, $m$, is substantially smaller then the starting dimension, $u$.
Under the assumption that $m = o(u^{1/2})$ the work in~\cite{Ailon08,Bamberger21} achieves embedding time $O(u \log m)$.
The only construction that for some regimes improves on the original Fast Johnson-Lindenstrauss Transform is the recent analysis~\cite{Jain22} of the Kac Johnson-Lindenstrauss Transform, which uses the Kac random walk~\cite{Kac56}.
They show that it can achieve an embedding time of $O(u \log u + \min\!\set{u \log n, m \log n \log^2(\log n) \log^3 u})$.

\paragraph*{Previous Work on Tabulation Hashing.}
The work by Patrascu and Thorup~\cite{Patrascu12} initiated the study of tabulation based hashing that goes further than what 3-wise independence of constructions would suggest.
A long line of papers have shown tabulation based hashing to work for min-wise hashing~\cite{Patrascu13,Dahlgaard14}, hashing for k-statistics~\cite{Dahlgaard15}, and the number of non-empty-bins~\cite{Aamand19}.
Furthermore, multiple papers have been concerned with showing strong concentration results for tabulation based hashing~\cite{Patrascu12,Patrascu13,Aamand20,Houen22}.
Tabulation based hashing has also been studied experimentally where they have been shown to exhibit great performance~\cite{Dahlgaard17,Aamand20}.

\subsection{Preliminaries}

In this section, we will introduce the notation which will be used throughout the paper.
First we introduce $p$-norms.
\begin{definition}[$p$-norm]
    Let $p \ge 1$ and $X$ be a random variable with $\ep{\abs{X}^p} < \infty$.
    We then define the $p$-norm of $X$ by $\pnorm{X}{p} = \ep{\abs{X}^p}^{1/p}$.
\end{definition}

Throughout the paper, we will repeatedly work with value functions $v \colon U \times [m] \to \R$.
We will allow ourself to sometime view them as vectors, and in particular, we will write
\begin{align*}
    \norm{v}{2} &= \sqrt{\sum_{x \in U} \sum_{j \in [m/s]} v(x, j)^2} , \\
    \norm{v}{\infty} &= \max_{x \in U, j \in [m/s]} \abs{v(x, j)} .
\end{align*}

We will also use the $\Psi_p$-function introduced in~\cite{Houen22}.
\begin{definition}
    For $p \ge 2$ we define the function $\Psi_p \colon \R_+ \times \R_{+} \to \R_{+}$ as follows,
    \begin{align*}
        \Psi_p(M, \sigma^2) = \begin{cases}
            \left(\frac{\sigma^2}{p M^2}\right)^{1/p} M &\text{if $p < \log \frac{p M^2}{\sigma^2}$} \\
            \tfrac{1}{2}\sqrt{p}\sigma &\text{if $p < e^{2} \frac{\sigma^2}{M^2}$} \\
            \frac{p}{e \log \frac{p M^2}{\sigma^2}} M &\text{if $\max\!\set{\log \frac{p M^2}{\sigma^2}, e^{2} \frac{\sigma^2}{M^2}} \le p$}
        \end{cases}
        \; .
    \end{align*}
\end{definition}
It was shown in~\cite{Houen22} that $\Psi_p(1, \lambda)$ is within a constant factor of the $p$-norm of a Poisson distributed random variable with parameter $\lambda$.
They also showed that $\Psi_p(M, \sigma^2)$ can be used to upper bound expressions involving a fully random hash function $h \colon U \to [m]$.
Let $v \colon U \times [m] \to \R$ be a value function then they showed that
\begin{align*}
    \pnorm{\sum_{x \in U} v(x, h(x))} \le C \Psi_p(\norm{v}{\infty}, \norm{v}{2}^2/m) \; ,
\end{align*}
where $C$ is a universal constant.

\section{Overview of the New Analysis}\label{sec:proof_overview}
Our main technical contribution is a new analysis of the Sparse Johnson-Lindenstrauss Transform that relaxes the assumptions on the hash function, $h$.
We show that if $h$ satisfies a decoupling decomposition property and a strong concentration property then we obtain the same bounds for the Sparse Johnson-Lindenstrauss Transform.
Both of these properties are satisfied by $h$ if $h$ is $\Omega(\log 1/\delta)$-wise independent so our assumptions are weaker than those of the previous analyses.

In this section, we will give an informal overview of new approach.
The technical details and the formal statement of the result will be in \Cref{sec:technical}.

In order to describe our approach, we look at the random variable
\begin{align}\label{eq:general-expression}
    Z 
        = \norm{Aw}{2}^2 - 1
        = \frac{1}{s} \sum_{i \in [s]} \sum_{x \neq y \in [u]} \sigma(i, x) \sigma(i, y) \indicator{h(i, x) = h(i, y)} w_x w_y
    .
\end{align}
Here $w \in \R^u$ is a unit vector.
With this notation the goal becomes to bound $\prb{\abs{Z} \ge \eps}$.

The first step in our analysis is that we want to decouple \cref{eq:general-expression}.
Decoupling was also used in one of the proofs in~\cite{Cohen18}, but since we want to prove the result for more general hash functions, we cannot directly use the standard decoupling inequalities.
We will instead assume that our hash function allows a \emph{decoupling-decomposition}.
This will formally be defined in \Cref{sec:technical} and we will for now assume that our hash function allows for the standard decoupling inequality.
If we apply Markov's inequality and a standard decoupling inequality for fully random hashing we obtain the expression.
\begin{align}\begin{split}\label{eq:decoupling-independence}
    \prb{\abs{Z} \ge \eps}
        &\le \eps^{-p} \ep{\abs{Z}^p}
        \\&\le \left(\eps^{-1} \frac{4}{s} \right)^p \ep{\abs{\sum_{i \in [s]} \sum_{x, y \in [u]} \sigma(i, x) \sigma'(i, y) \indicator{h(i, x) = h'(i, y)} w_x w_y}^p}
\end{split}\end{align}
where $(h', \sigma')$ are independent copies of $(h, \sigma)$ and $p \ge 2$.
The power of decoupling stems from the fact that it breaks up some of the dependencies and allows for a simpler analysis.

The goal is now to analyse $\pnorm{\sum_{i \in [s]} \sum_{x, y \in [u]} \sigma(i, x) \sigma'(i, y) \indicator{h(i, x) = h'(i, y)} w_x w_y}{p}$.
This is done by first fixing $(h', \sigma')$ and bounding $\pnorm{\sum_{i \in [s], j \in [m/s]} \sum_{x \in [u]} \sigma(i, x) \indicator{h(i, x) = j} w_x a_{i j}}{p}$ using the randomness of $(h, \sigma)$ where $a_{i j} = \sum_{y \in [u]} \sigma'(i, y) \indicator{h'(i, y) = j} w_y$.
In order to do this, we will assume that the pair $(h, \sigma)$ is \emph{strongly concentrated}.
Again the formal definition of this is postponed to \Cref{sec:technical}, but informally, we say that the pair is strongly concentrated if it has concentration results similar to those of fully random hashing.

We now take the view that $\abs{a_{i j}}$ is the load of the bin $(i, j) \in [s] \times [m/s]$.
The idea is then to split $[s] \times [m/s]$ into heavy and light bins and handle each separately.
We choose a parameter $k$ and let $I$ be the heaviest $k$ bins.
Using the triangle inequality, we then get that
\begin{align*}
    &\pnorm{\sum_{i \in [s], j \in [m/s]} \sum_{x \in [u]} \sigma(i, x) \indicator{h(i, x) = j} w_x a_{i j}}{p}
        \le \pnorm{\sum_{(i, j) \in I} \sum_{x \in [u]} \sigma(i, x) \indicator{h(i, x) = j} w_x a_{i j}}{p}
        \\&\qquad\qquad\qquad+ \pnorm{\sum_{(i, j) \in [s] \times [m/s] \setminus I} \sum_{x \in [u]} \sigma(i, x) \indicator{h(i, x) = j} w_x a_{i j}}{p}
    \; .
\end{align*}

We show that the contribution from the light bins is as if the collisions are independent.
This should be somewhat intuitive since if we only have few collisions in each bin then the collisions behave as if they were independent.
In contrast, we show that the contribution from the heavy bins is dominated by the heaviest bin.
This turns out to be exactly what we need to finish the analysis.

\section{Technical Results}\label{sec:technical}

In this section, we will expand on the description from \Cref{sec:proof_overview} and formalize the ideas.

\paragraph*{Decoupling.}
Ideally, we would like to use the standard decoupling inequality, \cref{eq:decoupling-independence}.
Unfortunately, we cannot expect more general hash functions to support such a clean decoupling.
We therefore introduce the notion of a decoupling-decomposition.

\begin{definition}[Decoupling-decomposition]
    Let $p \ge 2$, $L \ge 1$, and $0 \le \gamma \le 1$.
    We say that a collection of possibly randomized sets, $(U_\alpha)$, is a $(p, L, \gamma)$-\emph{decoupling-decomposition} for a property $P$ of a pair $(h, \sigma)$, if there exist hash functions $h_\alpha \colon [s] \times U_\alpha \to [m/s]$ and sign functions $s \colon [s] \times U_\alpha \to \set{-1, 1}$ for all $\alpha$ such that
    \begin{align}\begin{split}\label{eq:decoupling-decomposition}
        &\prb{\abs{Z} \ge \eps}
        \\&\qquad\le \left(\eps^{-1} \sum_{\alpha} \frac{L}{s} \pnorm{\sum_{i \in [s]}  \sum_{x, y \in U_\alpha} \sigma_\alpha(i, x) \sigma_\alpha'(i, y) \indicator{h_\alpha(i, x) = h_\alpha'(i, y)} w_x w_y}{p} \right)^p + \gamma
    \end{split}\end{align}
    where $(h_\alpha, \sigma_\alpha)$ and $(h'_\alpha, \sigma'_\alpha)$ has the same distribution, and $(h_\alpha, \sigma_\alpha)$ satisfies the property $P$ when conditioned on $(h'_\alpha, \sigma'_\alpha)$ and $U_\alpha$.

\end{definition}

The reader should compare \cref{eq:decoupling-independence} for fully random hashing with \cref{eq:decoupling-decomposition}.
There are 3 main differences between the expressions.
\begin{enumerate}
    \item The first thing to notice is that, in the decoupling-decomposition we sum over different sets $(U_\alpha)$, where this is not needed for fully random hashing.
        We allow the decoupling-decomposition to use a different decoupling on each of the sets $U_\alpha$.
        This is very powerful since general hash functions are not necessarily uniform over the input domain. 
    \item For the decoupling-decomposition, we allow an additive error probability $\gamma$.
        This is useful if the hash function allows for decoupling most of the time except when some unprobable event is happening.
    \item The last difference is that a much larger loss-factor is allowed by the decoupling-decomposition than \cref{eq:decoupling-independence}.
        In the case of fully random hashing, we only lose a factor of $4$ but for more general hash functions this loss might be bigger.
\end{enumerate}

Finally, we note that \cref{eq:decoupling-independence} implies if $(h, \sigma)$ is $2p$-wise independent for an integer $p \ge 2$ then $[u]$ is a decoupling-decomposition of $(h, \sigma)$ for any property $P$ that is satisfied by $(h, \sigma)$.

\paragraph*{Strong Concentration.}
The second property we need is that the hash function is strongly concentrated.

\begin{definition}[Strong concentration]
    Let $h \colon [s] \times U \to [m/s]$ be a hash function and $\sigma \colon [s] \times U \to \set{-1, 1}$ be a sign function.
    We say that the pair $(h, \sigma)$ is $(p, L)$-\emph{strongly-concentrated} if
    \begin{enumerate}
        \item For all value functions, $v \colon [s] \times [m/s] \to \R$, and all vectors, $w \in \R^U$,
        \begin{align}
            \label{eq:sjl-sum-poisson}
            \pnorm{\sum_{i \in [s]} \sum_{x \in U} \sigma(i, x) v(i, h(i, x)) w_x }{p} 
                &\le \Psi_{p}\left( L \norm{v}{\infty}\norm{w}{\infty} , L \frac{s}{m} \norm{v}{2}^2 \norm{w}{2}^2 \right) \; , \\
            \label{eq:sjl-sum-gauss}
            \pnorm{\sum_{i \in [s]} \sum_{x \in U} \sigma(i, x) v(i, h(i, x)) w_x }{p}
                &\le \sqrt{L \frac{p}{\log(m/s)}  \norm{v}{2}^2 \norm{w}{2}^2 } \; .
        \end{align}
        \item For all vectors, $w \in \R^U$,
        \begin{align}
        \label{eq:sjl-square}
            \pnorm{\sum_{i \in [s]} \sum_{j \in [m/s]} \left( \sum_{x \in U} \sigma(i, x) \indicator{h(i, x) = j} w_x \right)^2}{p/2}
                &\le L \max\!\set{s \norm{w}{2}^2, \frac{p}{\log m/s} \norm{w}{2}^2} .
        \end{align}
        \item If $p \le \log m$,
        \begin{align}\label{eq:sjl-max}
            \pnorm{\max_{i \in [s], j \in [m/s]} \abs{\sum_{x \in U} \sigma(i, x) \indicator{h(i, x) = j} w_x }}{p}
                \le e \sqrt{ L \frac{\log m}{\log m/s} } \norm{w}{2} .
        \end{align}
    \end{enumerate}
\end{definition}
We need essentially 3 different properties of our hash function to say that it is strongly concentrated.
\begin{enumerate}
    \item The first property is a concentration result on the random variable $$\sum_{i \in [s]} \sum_{x \in U} \sigma(i, x) v(i, h(i, x)) w_x.$$
        Here we need two different concentration results: The first concentration result, \cref{eq:sjl-sum-poisson}, roughly corresponds to a $p$-norm version of what you would obtain by applying Bennett's inequality to a fully random hash function, while the second concentration result, \cref{eq:sjl-sum-poisson}, corresponds to the best hypercontractive result you can obtain for weighted sums of independent Bernoulli-Rademacher variables with parameter $s/m$.\footnote{A Bernoulli-Rademacher variable with parameter $\alpha$ is random variable, $X \in \set{-1, 0, 1}$, with $\prb{X = 1} = \prb{X = -1} = \alpha/2$ and $\prb{X = 0} = 1 - \alpha$.}
    \item The second property bounds the sum of squares $$W = \sum_{i \in [s]} \sum_{j \in [m/s]} \left( \sum_{x \in U} \sigma(i, x) \indicator{h(i, x) = j} w_x \right)^2.$$
        The condition, \cref{eq:sjl-square}, bounds $\pnorm{W}{p/2}$ by the maximum of two cases.
        The first case corresponds to $\ep{W}$, and the second case is motivated by applying \cref{eq:sjl-sum-gauss} to
        \begin{align*}
            \sup_{\substack{z \in \R^{[s] \times [m/s]},\\ \norm{z}{2} = 1}} \pnorm{\sum_{i \in [s]} \sum_{j \in [m]} \sum_{x \in U} \sigma(i, x) z_{i, h(i, x)} w_x}{p}^2 .
        \end{align*}
        While this at first glance might seem odd, it is roughly the best you can do, since one can show that
        \begin{align*}
            \max\!\set{\ep{W}, \sup_{\substack{z \in \R^{[s] \times [m/s]},\\ \norm{z}{2} = 1}} \pnorm{\sum_{i \in [s]} \sum_{j \in [m]} \sum_{x \in U} \sigma(i, x) z_{i, h(i, x)} w_x}{p}^2}
                \le \pnorm{W}{p/2} .
        \end{align*}
    \item The final property is a bound on the largest coordinate, $\max_{i \in [s], j \in [m/s]} \abs{\sum_{x \in U} \sigma(i, x) \indicator{h(i, x) = j} w_x }$.
        The bound is a natural consequence of \cref{eq:sjl-sum-gauss} for fully random hashing.
        Namely, for fully random hashing we get that
        \begin{align*}
            &\pnorm{\max_{i \in [s], j \in [m/s]} \abs{\sum_{x \in U} \sigma(i, x) \indicator{h(i, x) = j} w_x }}{p}
                \\&\qquad\qquad\le \pnorm{\max_{i \in [s], j \in [m/s]} \abs{\sum_{x \in U} \sigma(i, x) \indicator{h(i, x) = j} w_x }}{\log m}
                \\&\qquad\qquad\le e \max_{i \in [s], j \in [m/s]}\pnorm{ \abs{\sum_{x \in U} \sigma(i, x) \indicator{h(i, x) = j} w_x }}{\log m}
                \\&\qquad\qquad\le e \sqrt{ L \frac{\log m}{\log m/s} } \norm{w}{2} .
        \end{align*}
        This derivation is not true for general hash function, but the hash function can still satisfy \cref{eq:sjl-max}.
\end{enumerate}

The results of~\cite{Houen22} show that if the hash function $h \colon [s] \times U \to [m/s]$ and the sign function $\sigma \colon [s] \times U \to [m/s]$ is $p$-wise independent for an integer $p \ge 2$ then the pair $(h, \sigma)$ is $(p, K)$-strongly-concentrated where $K$ is a universal constant.

\paragraph*{The Main Result.}
We are now ready to state our main result which is a new analysis of a Sparse Johnson-Lindenstrauss Transform that only assumes that the hash function has a decoupling-decomposition for the strong concentration property.

\begin{theorem}\label{thm:main}
    Let $h \colon [s] \times [u] \to [m/s]$ be a hash function and $\sigma \colon [s] \times [u] \to \set{-1, 1}$ be a sign function.
    Furthermore, let $0 < \eps < 1$ and $0 < \delta < 1$ be given, and define $p = \log 1/\delta$.

    Assume that there exists constants $L_1$, $L_2$, $L_3$, and $0 \le \gamma < 1$, that only depends on $(h, \sigma)$ and $p$, such that
    \begin{enumerate}
        \item There exists a $(p, L_1, \gamma)$-decoupling-decomposition, $(U_\alpha)$, for the $(p, L_2)$-strong-concentration property of $(h, \sigma)$
        \item For all vectors $w \in \R^u$, $\sum_\alpha \sum_{x \in U_\alpha} w_x^2 \le L_3 \norm{w}{2}^2$.
        \item $m \ge \left(16 e^7 L_1^2 L_2^3 L_3^2 \right) \cdot \eps^{-2} \log(1/\delta)$.
        \item $s \ge \left(64 e^3 L_1 L_2^{3/2} L_3 \right) \cdot \eps^{-1} \log(1/\delta)$.
    \end{enumerate}

    Then the following is true
    \begin{align*}
        \prb{\abs{Z} \ge \eps} \le \delta + \gamma .
    \end{align*}
\end{theorem}

As discussed earlier, a fully random hash function satisfies all the property needed of the theorem and thus gives a new analysis of the Sparse Johnson-Lindenstrauss Transform for fully random hashing.
We will also later show that Mixed Tabulation satisfies the assumption of the theorem hence giving the first analysis of a Sparse Johnson-Lindenstrauss Transform with a practical hash function that works.

The main difficulty in the analysis of \Cref{thm:main} is contained in the following technical lemma.
The idea in the proof of \Cref{thm:main} is to use the decoupling-decomposition and apply the following lemma to each part.

\begin{lemma}\label{lem:main-technical}
    Let $h, \overline{h} \colon [s] \times U \to [m/s]$ be hash functions and $\sigma, \overline{\sigma} \colon [s] \times U \to \set{-1, 1}$ be sign functions.
    Let $p \ge 2$ and assume that there exists a constant $L$ such that $(h, \sigma)$ is $(p, L)$-strongly concentrated when conditioning on $(\overline{h}, \overline{\sigma})$, and similarly, $(\overline{h}, \overline{\sigma})$ is $(p, L)$-strongly concentrated when conditioning on $({h}, {\sigma})$.
    Then for all vectors $w \in \R^U$,
    \begin{align*}
        &\pnorm{\sum_{i \in [s]} \sum_{x, y \in U} \sigma(i, x) \overline{\sigma}(i, y) \indicator{h(i, x) = \overline{h}(i, y)} w_x w_y}{p}
            \\&\qquad\le \Psi_p\left( 32 e^3 L^{3/2} \norm{w}{2}^2, 32 e^6 L^{3} \frac{s^2}{m} \norm{w}{2}^4 \right) 
                + 36 e^3 L \frac{p}{\log m/s} \norm{w}{2}^2 .
    \end{align*}
\end{lemma}

The lemma shows that the expression has two different regimes.
The first regime, $\Psi_p\left( 32 e^3 L^{3/2} \norm{w}{2}^2, 32 e^6 L^{3} \frac{s^2}{m} \norm{w}{2}^4 \right)$, is essentially what we would expect if each of the collisions, $\indicator{h(i, x) = \overline{h}(i, y)}$, are independent of each other.
The other regime, $36 e^3 L \frac{p}{\log m/s} \norm{w}{2}^2$, is essentially what you expect the largest coordinate to contribute.

Our analysis is inspired by these two regimes and tries to exploit them explicitly.
We start by fixing $(h, \sigma)$ and divide the coordinates into heavy and light coordinates.
We then show that contribution of the light coordinates is $\Psi_p\left( 32 e^3 L^{3/2} \norm{w}{2}^2, 32 e^6 L^{3} \frac{s^2}{m} \norm{w}{2}^4 \right)$ which matches the intuition that if we only have few collisions on each coordinate then the collisions behave as if they were independent.
Similarly, we show that the contribution of the heavy coordinates is dominated by the heaviest coordinate, namely, the contribution is $36 e^3 L \frac{p}{\log m/s} \norm{w}{2}^2$.

\subsection{Mixed Tabulation Hashing}\label{sec:mixed-results}

Our main result for Mixed Tabulation hashing is the following.
\begin{theorem}\label{thm:mixed-main}
    Let $h \colon [s] \times [u] \to [m/s]$ and $\sigma \colon [s] \times [u] \to \set{-1, 1}$ be Mixed Tabulation functions as described in \Cref{sec:construction}.
    Furthermore, let $0 < \eps < 1$ and $0 < \delta < 1$ be given, and define $p = \log 1/\delta$.

    If $m \ge \gamma_p^{3c} \eps^{-2} \log(1/\delta)$ and $s \ge \gamma_p^{3/2 c} \eps^{-1} \log(1/\delta)$
    where $\gamma_p = K c \max\!\set{1, \frac{p}{\log \abs{\Sigma}}}$ for a universal constant $K$.

    Then the following is true
    \begin{align*}
        \prb{\abs{Z} \ge \eps} \le \delta + \eps 3^c \abs{\Sigma}^{-d} .
    \end{align*}
\end{theorem}

The result follows by proving that Mixed Tabulation hashing has a $(p, 4^{c + 2}, 4 \eps^{-2} 3^c \frac{s}{m} \abs{\Sigma}^{-d})$-decoupling-decomposition and that Mixed Tabulation has the strong concentration property.
The main new part is in showing the decoupling-decomposition while the analysis of the strong concentration property is modification of the analysis in~\cite{Houen22}.

\section{Analysis of the Sparse Johnson-Lindenstrauss Transform}

Lets us start by showing how \Cref{lem:main-technical} implies our main result, \Cref{thm:main}.
\begin{proof}[Proof of~\Cref{thm:main}]
    We start by using \cref{eq:decoupling-decomposition} of the decoupling decomposition to get that
    \begin{align*}
        &\prb{\abs{Z} \ge \eps}
        \\&\qquad\qquad\le \left(\eps^{-1} \sum_{\alpha} \frac{L_1}{s} \pnorm{\sum_{i \in [s]}  \sum_{x, y \in U_\alpha} \sigma_\alpha(i, x) \sigma_\alpha'(i, y) \indicator{h_\alpha(i, x) = h_\alpha'(i, y)} v_x v_y}{p} \right)^p + \gamma
    \end{align*}
    Now we fix $\alpha$ and apply \Cref{lem:main-technical} while fixing $U_\alpha$
    \begin{align*}
        &\pnorm{\sum_{i \in [s]}  \sum_{x, y \in U_\alpha} \sigma_\alpha(i, x) \sigma_\alpha'(i, y) \indicator{h_\alpha(i, x) = h_\alpha'(i, y)} v_x v_y}{p}
            \\&\qquad\qquad\le \Psi_p\left(32 e^3 L_2^{3/2}, 32 e^6 L^3 \frac{s^2}{m} \right)\sum_{x \in U_\alpha} w_x^2
            + 36 e^3 L_2 \frac{p}{\log m/s} \sum_{x \in U_\alpha} w_x^2
    \end{align*}
    Using this we get that
    \begin{align*}
        &\sum_{\alpha} \frac{L_1}{s} \pnorm{\sum_{i \in [s]}  \sum_{x, y \in U_\alpha} \sigma_\alpha(i, x) \sigma_\alpha'(i, y) \indicator{h_\alpha(i, x) = h_\alpha'(i, y)} v_x v_y}{p}
            \\&\qquad\qquad\le \sum_{\alpha} \frac{L_1}{s} \left(\Psi_p\left(32 e^3 L_2^{3/2}, 32 e^6 L^3 \frac{s^2}{m} \right)\sum_{x \in U_\alpha} w_x^2
            + 36 e^3 L_2 \frac{p}{\log m/s} \sum_{x \in U_\alpha} w_x^2 \right)
    \end{align*}
    We now use that $\sum_\alpha \sum_{x \in U_\alpha} w_x^2 \le L_3 \norm{w}{2}^2$ to get that
    \begin{align*}
        &\sum_{\alpha} \frac{L_1}{s} \left(\Psi_p\left(32 e^3 L_2^{3/2}, 32 e^6 L^3 \frac{s^2}{m} \right)\sum_{x \in U_\alpha} w_x^2
            + 36 e^3 L_2 \frac{p}{\log m/s} \sum_{x \in U_\alpha} w_x^2 \right)
            \\&\qquad\qquad\le \frac{L_3 L_1}{s} \left(\Psi_p\left(32 e^3 L_2^{3/2}, 32 e^6 L^3 \frac{s^2}{m} \right)
            + 36 e^3 L_2 \frac{p}{\log m/s} \right) \norm{w}{2}^2
    \end{align*}
    It can now be checked that if $m$ and $s$ satisfies the stated assumptions then
    \begin{align*}
        \frac{L_3 L_1}{s} \left(\Psi_p\left(32 e^3 L_2^{3/2}, 32 e^6 L^3 \frac{s^2}{m} \right)
            + 36 e^3 L_2 \frac{p}{\log m/s} \right) \norm{w}{2}^2
            \le e^{-1} \eps
    \end{align*}

    Combining all the facts, we get that
    \begin{align*}
        \prb{\abs{Z} \ge \eps}
        \le \left(\eps^{-1} (e^{-1} \eps)\ \right)^p + \gamma
        = \delta + \gamma .
    \end{align*}
    This finishes the proof.
\end{proof}

The rest of the section is concerned with proving our main technical lemma, \Cref{lem:main-technical}.
First we need the following two lemmas from~\cite{Houen22}.
\begin{lemma}\label{lem:fn-moment}
    Let $f \colon \R_{\ge 0}^n \to \R_{\ge 0}$ be a non-negative function which is monotonically increasing in every argument, and assume that there exists positive reals $(\alpha_i)_{i \in [n]}$ and $(t_i)_{i \in [n]}$ such that for all $\lambda \ge 0$,
    \begin{align*}
        f(\lambda^{\alpha_0} t_0, \ldots, \lambda^{\alpha_{n - 1}} t_{n - 1})
            \le \lambda f(t_0, \ldots, t_{n - 1})
        \; .
    \end{align*}
    Let $(X_i)_{i \in [n]}$ be non-negative random variables.
    Then for all $p \ge 1$ we have that
    \begin{align*}
        \pnorm{f(X_0, \ldots, X_{n - 1})}{p}
            \le n^{1/p} \max_{i \in [n]} \left( \frac{\pnorm{X_i}{p/\alpha_i}}{t_i} \right)^{1/\alpha_i} f(t_0, \ldots, t_{n - 1})
        \; .
    \end{align*}
\end{lemma}
\begin{lemma}\label{lem:psi-properties}
    Let $p \ge 2$, $M > 0$, and $\sigma^2 > 0$ then
    \begin{align*}
        \frac{1}{2}\sqrt{p}\sigma \le \Psi_p(M, \sigma^2) \le \max\!\set{\frac{1}{2}\sqrt{p}\sigma, \frac{1}{2e}p M} .
    \end{align*}
\end{lemma}

We are now ready to prove \Cref{lem:main-technical}.
\begin{proof}[Proof of \Cref{lem:main-technical}]
    We start by defining $v_h, v_{\bar{h}} \colon [s] \times [m/s] \to \R$ by,
    \begin{align*}
        v_h(i, j) &= \sum_{x \in U} \sigma(i, x) w_x \indicator{h(i, x) = j} \; , \\
        v_{\bar{h}}(i, j) &= \sum_{y \in U} \overline{\sigma}(i, y) w_y \indicator{\overline{h}(i, y) = j} \; .
    \end{align*}
    We then want to prove that
    \begin{align*}
        &\pnorm{\sum_{i \in [s], j \in [m/s]} v_h(i, j) v_{\bar{h}}(i, j) }{p}
            \\&\qquad\qquad\le \Psi_p\left( 32 e^3 L^{3/2} \norm{w}{2}^2, 32 e^6 L^{3} \norm{w}{2}^4 \right) + 4 e^3 L \frac{p}{\log m/s} \norm{w}{2}^2
        \; .
    \end{align*}

    First we consider the case where $\frac{p}{\log m/s} \norm{w}{2}^2 \ge  s \norm{w}{2}^2$.
    By Cauchy-Schwartz and \cref{eq:sjl-square} we get that
    \begin{align*}
        \pnorm{\sum_{i \in [s], j \in [m/s]} v_h(i, j) v_{\bar{h}}(i, j) }{p}
            \le \pnorm{\sum_{i \in [s], j \in [m/s]} v_h(i, j)^2}{p}
            \le  L \frac{p}{\log m/s} \norm{w}{2}^2
        \; .
    \end{align*}

    We now focus on the case where $\frac{p}{\log m/s} \norm{w}{2}^2 < s \norm{w}{2}^2$.
    We define $\pi \colon [m] \to [s] \times [m / s]$ to be a bijection which satisfies that
    \begin{align*}
        \abs{v_h(\pi(0))} \ge \abs{v_h(\pi(1))} \ge \ldots \ge \abs{v_h(\pi(m - 1))} \; .
    \end{align*}
    We note that $\pi$ is a random function but we can define $\pi$ such that it only depends on the randomness of $h$ and $\sigma$.
    We define $k = \floor{p/\log(m/p)}$, $I = \setbuilder{\pi(i)}{i \in [k]}$, and the random functions $v'_h, v''_h \colon [s] \times [m/s] \to \R$ by
    \begin{align*}
        v'_h(i, j) &= v_h(i, j) \indicator{(i, j) \in I} , \\
        v''_h(i, j) &= v_h(i, j) \indicator{(i, j) \not\in I} . 
    \end{align*}
    Again we note that $v'_h$ and $v''_h$ only depends on the randomness of $h$ and $\sigma$.
    We can then write our expression as
    \begin{align*}
        &\pnorm{\sum_{i \in [s], j \in [m/s]} v_h(i, j) v_{\bar{h}}(i, j) }{p}
            = \pnorm{\sum_{i \in [s]} \sum_{y \in U} \overline{\sigma}(i, y) v_h(i, \overline{h}(i, y)) w_y }{p}
            \\&\qquad\qquad\le \pnorm{\sum_{i \in [s]} \sum_{y \in U} \overline{\sigma}(i, y) v'_h(i, \overline{h}(i, y)) w_y }{p}
            + \pnorm{\sum_{i \in [s]} \sum_{y \in U} \overline{\sigma}(i, y) v''_h(i, \overline{h}(i, y)) w_y }{p}
        .
    \end{align*}
    We will bound each of the term separately.
    We start by bounding $\pnorm{\sum_{i \in [s]} \sum_{y \in U} \overline{\sigma}(i, y) v'_h(i, \overline{h}(i, y)) w_y }{p}$.
    We fix $h$ and $\sigma$ and use \cref{eq:sjl-sum-gauss} to get that
    \begin{align*}
        \pnorm{\sum_{i \in [s]} \sum_{y \in U} \overline{\sigma}(i, y) v'_h(i, \overline{h}(i, y)) w_y }{p}
            &\le \pnorm{\sqrt{ L \frac{p}{\log m/s}  \norm{w}{2}^2 \norm{v'_h}{2}^2 }}{p}
            \\&= \sqrt{ L\frac{p}{\log m/s} } \norm{w}{2} \pnorm{\sqrt{ \sum_{(i, j) \in I} v'_h(i, j)^2 }}{p}
        .
    \end{align*}
    We note that $\sum_{(i, j) \in I} v'_h(i, j)^2 = \max_{J \subseteq [s] \times [m/s], |J| = k} \sum_{(i, j) \in J} v_h(i, j)^2$.
    We then get that
    \begin{align*}
        \pnorm{\sqrt{ \sum_{(i, j) \in I} v'_h(i, j)^2 }}{p}
            &= \pnorm{\sqrt{ \max_{J \subseteq [s] \times [m/s], |J| = k} \sum_{(i, j) \in J} v_h(i, j)^2 }}{p}
            \\&\le \left( \sum_{J \subseteq [s] \times [m/s], |J| = k} \pnorm{\sqrt{\sum_{(i, j) \in J} v_h(i, j)^2}}{p}^p \right)^{1/p}
            \\&\le \binom{ms}{k}^{1/p} \max_{J \subseteq [s] \times [m/s], |J| = k} \pnorm{\sqrt{\sum_{(i, j) \in J} v_h(i, j)^2}}{p}
    \end{align*}
    We use Sterling's bound and get that $\binom{ms}{k}^{1/p}
    \le \left( \frac{e ms}{k} \right)^{k/p}
    \le \left( \frac{e ms \log(ms/p)}{p} \right)^{1/\log(ms/p)}
    \le e^3$.
    So we get that
    \begin{align*}
        \pnorm{\sqrt{ \sum_{(i, j) \in I} v'_h(i, j)^2 }}{p}
            \le e^3 \max_{J \subseteq [s] \times [m/s], |J| = k} \pnorm{\sqrt{\sum_{(i, j) \in J} v_h(i, j)^2}}{p}
    \end{align*}
    A standard volumetric argument gives that there exists a $1/4$-net, $Z \subseteq \R^J$, with $\abs{Z} \le 9^k$, such that
    \begin{align*}
        \pnorm{\sqrt{\sum_{(i, j) \in J} v_h(i, j)^2}}{p}
            &= \pnorm{ \sup_{z \in \R^{J}, \norm{z}{2} = 1} \sum_{(i, j) \in J} z_{i, j} v_h(i, j)}{p}
            \\&\le \pnorm{ \sup_{z \in Z} \sum_{(i, j) \in J} z_{i, j} v_h(i, j) }{p}
            \\&\qquad+ \pnorm{\sup_{z \in \R^{J}, \norm{z}{2} = 1} \sum_{(i, j) \in J} (z_{i, j} - z'_{i, j}) v_h(i, j) }{p}
            \\&\le \pnorm{ \sup_{z \in Z} \sum_{(i, j) \in J} z_{i, j} v_h(i, j) }{p}
            \\&\qquad+ \sup_{z \in \R^{J}, \norm{z}{2} = 1} \norm{z - z'}{2} \pnorm{\sqrt{\sum_{(i, j) \in J} v_h(i, j)^2}}{p}
    \end{align*}
    where $z' \in Z$ is the closest element to $z$, and as such $\norm{z - z'}{2} \le 1/4$.
    Since there are at most $9^k$ elements in $Z$ then $\pnorm{ \sup_{z \in Z} \sum_{(i, j) \in J} z_{i, j} v_h(i, j) }{p} \le 9 \sup_{z \in Z} \pnorm{ \sum_{(i, j) \in J} z_{i, j} v_h(i, j) }{p}$, where we used that $k \le p$.
    Collecting the fact we get that
    \begin{align*}
        \pnorm{\sqrt{\sum_{(i, j) \in J} v_h(i, j)^2}}{p}
            \le 36 \sup_{z \in Z} \pnorm{ \sum_{(i, j) \in J} z_{i, j} v_h(i, j) }{p}
    \end{align*}
    Using this we get that
    \begin{align*}
        &e^3 \max_{J \subseteq [s] \times [m/s], |J| = k} \pnorm{\sqrt{\sum_{(i, j) \in J} v_h(i, j)^2}}{p}
            \\&\qquad\qquad\le 36 e^3 \max_{J \subseteq [s] \times [m/s], |J| = k} \max_{z \in Z} \pnorm{\sum_{(i, j) \in J} z_{i, j} v_h(i, j)}{p}
            \\&\qquad\qquad= 36 e^3 \max_{J \subseteq [s] \times [m/s], |J| = k} \max_{\substack{z \in \R^{s \times m/s},\\ \norm{z}{2} = 1}} \pnorm{\sum_{i \in [s]} \sum_{x \in U} \sigma(i, x) z_{i, h(i, x)} \indicator{(i, h(i, x)) \in J} w_x}{p}
    \end{align*}
    We can then use \cref{eq:sjl-sum-gauss} to get that
    \begin{align*}
        &36 e^3 \max_{J \subseteq [s] \times [m/s], |J| = k} \max_{\substack{z \in \R^{s \times m/s},\\ \norm{z}{2} = 1}} \pnorm{\sum_{i \in [s]} \sum_{x \in U} \sigma(i, x) z_{i, h(i, x)} \indicator{(i, h(i, x)) \in J} w_x}{p}
            \\&\qquad\qquad\qquad\le 36 e^3  \max_{J \subseteq [s] \times [m/s], |J| = k} \max_{\substack{z \in \R^{s \times m/s},\\ \norm{z}{2} = 1}} \sqrt{L \frac{p}{\log m/s}} \norm{w}{2} \norm{z}{2}
            \\&\qquad\qquad\qquad= 36 e^3 \sqrt{L \frac{p}{\log m/s}} \norm{w}{2}
    \end{align*}
    Combining the facts, we get that
    $
        \pnorm{\sum_{i \in [s]} \sum_{y \in U} \overline{\sigma}(i, y) v'_h(i, \overline{h}(i, y)) w_y }{p}
            \le 36 e^3 L \frac{p}{\log m/s} \norm{w}{2}^2
    $.

    We will now bound $\pnorm{\sum_{i \in [s]} \sum_{y \in U} \overline{\sigma}(i, y) v''_h(i, \overline{h}(i, y)) w_y }{p}$.
    We fix $h$ and $\eps$ and use \cref{eq:sjl-sum-poisson} to get that
    \begin{align*}
        \pnorm{\sum_{i \in [s]} \sum_{y \in U} \overline{\sigma}(i, y) v''_h(i, \overline{h}(i, y)) w_y }{p}
            &\le \pnorm{\Psi_{p}\left( L \norm{w}{\infty} \norm{v'_h}{\infty}, L \frac{s}{m} \norm{w}{2}^2 \norm{v'_h}{2}^2 \right)}{p}
            \\&\le \pnorm{\Psi_{p}\left( L \norm{w}{\infty} \abs{v'_h(\pi(k + 1))}, L \frac{s}{m} \norm{w}{2}^2 \norm{v_h}{2}^2 \right)}{p}
        \; .
    \end{align*}
    Now we use \Cref{lem:fn-moment} to get that,
    \begin{align*}
        &\pnorm{\Psi_{p}\left( L \norm{w}{\infty} \abs{v'_h(\pi(k + 1))}, L \frac{s}{m} \norm{w}{2}^2 \norm{v_h}{2}^2 \right)}{p}
            \\&\qquad\qquad\le \sqrt{2} \Psi_{p}\left( L \norm{w}{\infty} \pnorm{\abs{v'_h(\pi(k + 1))}}{p}, L \frac{s}{m} \norm{w}{2}^2 \pnorm{\norm{v_h}{2}^2}{p/2} \right)
        .
    \end{align*}
    Since we assume that $\frac{p}{\log m} \norm{w}{2}^2 < s \norm{w}{2}^2$ then \cref{eq:sjl-square} give us that $\pnorm{\norm{v_h}{2}^2}{p/2} \le L s \norm{w}{2}^2$.
    
    We will now bound $\pnorm{v'_h(\pi(k + 1))}{p}$.
    For this, we will distinguish between two cases: Either $p \ge \log m$ or $p < \log m$.
    Let us first case where $p \ge \log m$.
    We will use that $\abs{v'_h(\pi(k + 1))} \le \frac{\sum_{i \in [k + 1]} \abs{v'_h(\pi(i))} }{k + 1}$.
    We then get that
    \begin{align*}
        &\pnorm{v'_h(\pi(k + 1))}{p}
            \\&\le \pnorm{\frac{\sum_{i \in [k + 1]} \abs{v'_h(\pi(i))} }{k + 1}}{p}
            \\&\le \left( \binom{m}{k + 1} 2^{k + 1} \max_{\substack{J \subseteq [s] \times [m/s], \\ |J| = k + 1}} \max_{(\sigma_{i, j})_{(i, j) \in J} \in \set{-1, 1}^J} \left( \frac{\pnorm{\sum_{(i, j) \in J} \sigma_{i, j} v_{h}(i, j)}{p}}{k + 1} \right)^p \right)^{1/p}
            \\&\le \max_{\substack{J \subseteq [s] \times [m/s], \\ |J| = k + 1}} \max_{(s_{i, j})_{(i, j) \in J} \in \set{-1, 1}^J} 2 \binom{m}{k + 1}^{1/p} \frac{\pnorm{\sum_{(i, j) \in J} \sigma_{i, j} v_{h}(i, j)}{p}}{k + 1}
    \end{align*}
    We note that $\pnorm{\sum_{(i, j) \in J} \sigma_{i, j} v_{h}(i, j)}{p} = \pnorm{\sum_{x \in U} \sum_{(i, j) \in J}   \sigma(i, x)  s_{i, j} \indicator{h(i, x) = j} w_x }{p}$.
    Since we have that $p \ge \log m$ then $k \ge 1$ which implies that $k + 1 \le 2\frac{p}{\log(m/p)}$.
    We then get that $\binom{m}{k + 1}^{1/p}
    \le \left( \frac{e m}{2 p/\log(m/p)}   \right)^{2/\log(m/p)}
    \le 2 e^3$.
    We now use \cref{eq:sjl-sum-gauss} to get that,
    \begin{align*}
        \pnorm{\sum_{x \in U} \sum_{(i, j) \in J}  \sigma(i, x)  s_{i, j} \indicator{h(i, x) = j} w_x}{p}
            &= \pnorm{\sum_{i \in [s]} \sum_{x \in U} \sigma(i, x) \indicator{(i, h(i, x)) \in J} s_{i, h(i, x)} w_x}{p}
            \\&\le \sqrt{L \frac{p}{\log m/s}} \sqrt{\abs{J}} \norm{w}{2}
            \\&= \sqrt{L \frac{p}{\log m/s}} \sqrt{k + 1} \norm{w}{2}
    \end{align*}
    Combining this we get that
    $
        \pnorm{v'_h(\pi(k + 1))}{p}
            \le 4 e^3 \sqrt{L \frac{p}{\log m/s}} \frac{\norm{w}{2}}{\sqrt{k + 1}}
    $.
    We then obtain that,
    \begin{align*}
        &\pnorm{\sum_{i \in [s]} \sum_{y \in U} \overline{\sigma}(i, y) v''_h(i, \overline{h}(i, y)) w_y }{p}
            \\&\qquad\qquad\le \sqrt{2} \Psi_p\left( 4 e^3 L \sqrt{L \frac{p}{(k + 1) \log m/s}} \norm{w}{\infty} \norm{w}{2}, L^2 \frac{s^2}{m} \norm{w}{2}^4 \right)
            \\&\qquad\qquad\le \sqrt{2} \Psi_p\left( 4 e^3 L \sqrt{L \frac{\log m/p}{ \log m/s}} \norm{w}{2}^2, L^2 \frac{s^2}{m} \norm{w}{2}^4 \right)
    \end{align*}
    If $\log m/p \le 4 \log m/s$ then we get that,
    \begin{align*}
        \pnorm{\sum_{i \in [s]} \sum_{y \in U} \overline{\sigma}(i, y) v''_h(i, \overline{h}(i, y)) w_y }{p}
            &\le \sqrt{2} \Psi_p\left( 16 e^3 L^{3/2} \norm{w}{2}^2, L^2 \frac{s^2}{m} \norm{w}{2}^4 \right)
            \\&\le \Psi_p\left( 32 e^3 L^{3/2} \norm{w}{2}^2, 2 L^2 \frac{s^2}{m} \norm{w}{2}^4 \right)
    \end{align*}
    If $\log m/p > 4 \log m/s$ then $m/p > (m/s)^4$ which implies that $\frac{p m}{s^2} \le \sqrt{p}{m}$.
    Using this we get that $
        \frac{p 16 e^6 L^3 \frac{\log m/p}{\log m/s} \norm{w}{2}^4}{L^2 \frac{s^2}{m} \norm{w}{2}^4}
            \le 16 e^6 L \sqrt{p}{m}\log m/p
            \le 16 e^6 L
    $.
    Where we have used that $\sqrt{s} \log 1/x \le 1$.
    Now we use \Cref{lem:psi-properties} to get that
    \begin{align*}
        \pnorm{\sum_{i \in [s]} \sum_{y \in U} \overline{\sigma}(i, y) v''_h(i, \overline{h}(i, y)) w_y }{p}
            &\le \sqrt{2} \Psi_p\left( 32 e^3 L^{3/2} \norm{w}{2}^2, L^2 \frac{s^2}{m} \norm{w}{2}^4 \right)
            \\&\le \sqrt{2} \sqrt{p L^3 16 e^6 \frac{s^2}{m} \norm{w}{2}^4}
            \\&\le \Psi_p\left(8 e^3 L^{3/2} \norm{w}{2}, 32 e^6 L^3 \frac{s^2}{m} \norm{w}{2}^4\right)
    \end{align*}

    Now let us consider the case where $p < \log m$.
    By \cref{eq:sjl-max}, we get that
    \begin{align*}
        \pnorm{v'_h(\pi(k + 1))}{p}
            \le \pnorm{\max_{i \in [s], j \in [m/s]} \abs{\sum_{x \in U} \sigma(i, x) \indicator{h(i, x) = j} w_x}}{p}
            \le e \sqrt{ L \frac{\log m}{\log m/s}} \norm{w}{2}
    \end{align*}
    We then obtain that,
    \begin{align*}
        \pnorm{\sum_{i \in [s]} \sum_{y \in U} \overline{\sigma}(i, y) v''_h(i, \overline{h}(i, y)) w_y }{p}
            &\le \sqrt{2} \Psi_p\left( e L \sqrt{L \frac{\log m}{\log m/s}} \norm{w}{\infty} \norm{w}{2}, L \frac{s^2}{m} \norm{w}{2}^4 \right)
            \\&\le \sqrt{2} \Psi_p\left( e L \sqrt{L \frac{\log m}{\log m/s}} \norm{w}{2}^2, L \frac{s^2}{m} \norm{w}{2}^4 \right)
    \end{align*}
    If $s \le m^{3/4}$ then we get that $\log m \le 4 \log m/s$ and
    \begin{align*}
        \sqrt{2} \Psi_p\left( e L \sqrt{L \frac{\log m}{\log m/s}} \norm{w}{2}^2, L \frac{s^2}{m} \norm{w}{2}^4 \right)
            \le \sqrt{2} \Psi_p\left( 4 e L^{3/2} \norm{w}{2}^2, L \frac{s^2}{m} \norm{w}{2}^4 \right)
    \end{align*}
    as wanted.
    If $s \ge m^{3/4}$ then we get that
    $
        \frac{p e L^{3} \log m}{L^2 s^2/m}
            \le \frac{e L m \log^2 m}{s^2}
            \le \frac{e L \log^2 m}{m^{1/2}}
            \le 16/e L
    $,
    where we have used that $\sqrt{x} \log^2 1/x \le 16/e^2$.
    Again we use \Cref{lem:psi-properties} to get that
    \begin{align*}
        \sqrt{2} \Psi_p\left( 4 e L^{3/2} \norm{w}{2}^2, L \frac{s^2}{m} \norm{w}{2}^4 \right)
            &\le \sqrt{p 32/e L^3 \frac{s^2}{m}} \norm{w}{2}^2
            \\&\le \Psi_p\left( 8 L^{3/2} \norm{w}{2}^2, 32 L^{3} \frac{s^2}{m} \norm{w}{2}^4 \right) .
    \end{align*}

    Combining everything we get that
    \begin{align*}
        \pnorm{\sum_{i \in [s]} \sum_{y \in U} \overline{\sigma}(i, y) v''_h(i, \overline{h}(i, y)) w_y }{p}
            &\le \Psi_p\left(32 e^3 L^{3/2} \norm{w}{2}, 32 e^6 L^3 \frac{s^2}{m} \norm{w}{2}^4\right) , \\
        \pnorm{\sum_{i \in [s]} \sum_{y \in U} \overline{\sigma}(i, y) v'_h(i, \overline{h}(i, y)) w_y }{p}
            &\le 4 e^3 L \frac{p}{\log m/s} \norm{w}{2}^2 .
    \end{align*}
    Now we conclude that 
    \begin{align*}
        &\pnorm{\sum_{i \in [s], j \in [m/s]} v_h(i, j) v_{\bar{h}}(i, j) }{p}
            \\&\qquad\qquad\le \pnorm{\sum_{i \in [s]} \sum_{y \in U} \overline{\sigma}(i, y) v'_h(i, \overline{h}(i, y)) w_y }{p}
            + \pnorm{\sum_{i \in [s]} \sum_{y \in U} \overline{\sigma}(i, y) v''_h(i, \overline{h}(i, y)) w_y }{p}
            \\&\qquad\qquad\le \Psi_p\left( 32 e^3 L^{3/2} \norm{w}{2}^2, 32 e^6 L^{3} \norm{w}{2}^4 \right) + 4 e^3 L \frac{p}{\log m/s} \norm{w}{2}^2 .
    \end{align*}
    Thus finishing the proof.
\end{proof}

\section{Analysis of Mixed Tabulation Hashing}\label{sec:mixed-analysis}

The goal of this section is to prove our main result for Mixed Tabulation hashing.
The main new results is in \Cref{sec:mixed-decoupling} where we show that Mixed Tabulation has a decoupling-decomposition.
In \Cref{sec:mixed-concentration}, we show that Mixed Tabulation also has the strong concentration property.
The proofs in \Cref{sec:mixed-concentration} are modifications of those found in~\cite{Houen22} for Mixed Tabulation hashing.

\subsection{Notation and Previous Results for Tabulation Hashing}
We will need to reason about the individual characters of a key, $x \in \Sigma^c$, and for that, we need some notation.

\begin{definition}[Position characters]
    Let $\Sigma$ be an alphabet and $c > 0$ a positive integer.
    We call an element $(i, y) \in [c] \times \Sigma$ a position character of $\Sigma^c$.
\end{definition}

We will view a key $x = (y_0, \ldots, y_{c - 1}) \in \Sigma^c$ as a set of $c$ position characters, $\set{(0, y_0), \ldots (c - 1, y_{c - 1})} \subseteq [c] \times \Sigma$.
We define the sets $P_{partial}$ and $P_{prefix}$ which contains partial keys.
\begin{align*}
    P_{partial} &= \setbuilder{\set{(i, \alpha_i)}_{i \in I}}{\emptyset \neq I \subseteq [c], \forall i \in I : \alpha_i \in \Sigma} \\
    P_{prefix} &= \setbuilder{\set{(i, \alpha_i)}_{i \in I}}{\emptyset \neq I \subseteq [c],\text{ $I$ is an interval containing $0$, and } \forall i \in I : \alpha_i \in \Sigma}
\end{align*}
For a partial key, $\pi = \set{(i_0, \alpha_0), \ldots, (i_{k - 1}, \alpha_{k - 1})} \in P_{partial}$, we define $I_{\pi} = \set{i_0, \ldots, i_{k - 1}}$ to be the set of used positions.
We will also write $\abs{\pi} = \abs{I_{\pi}}$.

For a simple tabulation hashing function, $h \colon \Sigma^c \to R$, we will let $h_I$ for $I \subseteq [c]$ to be the hash function that only looks at the positions in $I$, i.e., $h_I(x) = \bigxor_{i \in I} T_i[x_i]$.
Similarly, for $h_2 \colon \Sigma^c \to \Sigma^d$ we will define $h^{I}$ for $I \subseteq [d]$ to be the partial key restricted to the positions $I$.

For $\pi \in P_{partial}$ and $\tau \in P_{prefix}$, we define the random set $U_{\pi, \tau}$ as follows
\begin{align*}
    U_{\pi, \tau} = \setbuilder{x \in \Sigma^c}{\pi \subseteq x, \tau \subseteq h_2(x)}
\end{align*}

We also need the following lemmas from~\cite{Houen22}.

\begin{lemma}\label{lem:fully-random-concentration}
    Let $h \colon U \to [m]$ be a uniformly random function, let $v \colon U \times [m] \to \R$ be a fixed value function, and
    assume that $\sum_{j \in [m]} v(x, j) = 0$ for all keys $x \in U$.
    Define the random variable $X_v = \sum_{x \in U} v(x, h(x))$.
    Then for all $p \ge 2$,
    \begin{align*}
        \pnorm{X_v}{p}
            \le L \Psi_p\left(M_v, \sigma_v^2 \right)
        \; ,
    \end{align*}
    where $L \le 16 e$ is a universal constant.
\end{lemma}

\begin{lemma}\label{lem:fully-random-gauss}
    Let $h \colon U \to [m]$ be a uniformly random function, let $\eps \colon U \to \set{-1, 1}$ be a uniformly random sign function, and let $v : U \times [m] \to \R$ be a fixed value function.
    Then for all $p \ge 2$,
    \begin{align*}
        \pnorm{\sum_{x \in U} \eps(s) v(x, h(x))}{p}
            \le L \sqrt{\frac{p}{\log\!\left( m \right)}} \norm{v}{2}
        \; ,
    \end{align*}
    where $L \le e$ is a universal constant.
\end{lemma}

\begin{lemma}\label{lem:simple-concentration}
    Let $h \colon \Sigma^c \to [m]$ be a simple tabulation hash function, $v \colon \Sigma^{c} \times [m] \to \R$ a value function, and assume that $\sum_{j \in [m]} v(x, j) = 0$ for all keys $x \in U$.
    Then for all $p \ge 2$,
    \begin{align*} 
        \pnorm{\sum_{x \in \Sigma^c} v(x, h(x))}{p}
            \le L_1 \Psi_p\left(K_c \gamma_p^{c - 1} \norm{v}{\infty}, K_c \gamma_p^{c - 1} \frac{\norm{v}{2}^2}{m}  \right)
        \; ,
    \end{align*}
    where $K_c = \left(L_2 c \right)^{c - 1}$, $L_1$ and $L_2$ are universal constants, and 
    \begin{align*}
        \gamma_p = \frac{\max\!\set{\log(m) + \log\!\left(\tfrac{\sum_{x \in \Sigma^c} \norm{v[x]}{2}^2}{\max_{x \in \Sigma^c} \norm{v[x]}{2}^2}\right)/c, p}}
            {\log\!\left( e^2 m \left( \max_{x \in \Sigma^c} \frac{\norm{v[x]}{1}^2}{\norm{v[x]}{2}^2} \right)^{-1} \right)}
    \end{align*}
\end{lemma}

\begin{lemma}\label{lem:simple-squares}
    Let $h \colon \Sigma^c \to [m]$ be a simple tabulation function, $\eps \colon \Sigma^c \to \set{-1, 1}$ be a simple tabulation sign function, and $v_i \colon \Sigma^c \times [m] \to \R$ be value function for $i \in [k]$.
    For every real number $p \ge 2$,
    \begin{align*}\begin{split}
        &\pnorm{\sum_{j \in [m]} \left( \sum_{x \in \Sigma^c} \eps(x) v(x, h(x) \xor j) \right)^2 }{p}
            \le \left( \frac{L c \max\!\set{p, \log(m)}}{\log\!\left( \frac{e^2 m \sum_{x \in \Sigma^c} \sum_{j \in [m]} v(x, j)^2}{\sum_{x \in \Sigma^c} (\sum_{j \in [m]} \abs{v(x, j)} )^2} \right)} \right)^{c} \norm{v}{2}^2 
        \; ,
    \end{split}\end{align*}
    where $L$ is a universal constant.
\end{lemma}

\subsection{Decoupling}\label{sec:mixed-decoupling}

Before proving the decoupling-decomposition for Mixed Tabulation hashing, we a general decoupling lemma.
We start by stating a standard decoupling result as found in~\cite{hdpBook}.
\begin{lemma}\label{lem:decoupling-1d}
    Let $f_{i j} \colon T \times T \to \R$ be functions for $i, j \in [n]$.
    Let $(X_i)_{i \in [n]}$ be independent random variables with values in $T$ such that $\epcond{f_{ij}(X_i, X_j)}{X_j} = \epcond{f_{ij}(X_i, X_j)}{X_i} = 0$ for all $i \neq j \in [n]$.
    Then, for every convex function $F \colon \R \to \R$, one has
    \begin{align*}
        \ep{F(\sum_{i \neq j \in [n]} f_{ij}(X_i, X_j))}
            \le \ep{F(4 \sum_{i, j \in [n]} f_{ij}(X_i, X'_j))} \; ,
    \end{align*}
    where $(X'_i)_{i \in [n]}$ is an independent copy of $(X_i)_{i \in [n]}$.
\end{lemma}

We slightly generalize the decoupling result as follows.

\begin{lemma}\label{lem:decoupling-gen}
    Let $f_{i_0, \ldots, i_{c - 1}, j_0, \ldots, j_{c - 1}} \colon T^c \times T^c \to \R$ be functions for $i_0 \ldots, i_{c - 1}, j_0, \ldots, j_{c - 1} \in [n]$.
    Let $(X^{(k)}_i)_{i \in [n], k \in [c]}$ be independent random variables with values in $T$ such that 
    \begin{align}\label{eq:decoupling-gen-cond}
    \begin{split}            
        \epcond{f_{i_0, \ldots, i_{c - 1}, j_0, \ldots, j_{c - 1}}((X^{(0)}_{i_0}, \ldots, X^{(c - 1)}_{i_{c - 1}}), (X^{(0)}_{j_0}, \ldots, X^{(c - 1)}_{j_{c - 1}}))}{(X^{(k)}_{i_k})_{k \in [c] \setminus \set{l}}, (X^{(k)}_{j_k})_{k \in [c]}} 
        &= 0 \; , \\
        \epcond{f_{i_0, \ldots, i_{c - 1}, j_0, \ldots, j_{c - 1}}((X^{(0)}_{i_0}, \ldots, X^{(c - 1)}_{i_{c - 1}}), (X^{(0)}_{j_0}, \ldots, X^{(c - 1)}_{j_{c - 1}}))}{(X^{(k)}_{i_k})_{k \in [c]}, (X^{(k)}_{j_k})_{k \in [c] \setminus \set{l}}} 
        &= 0 \; ,
    \end{split}
    \end{align}
    for all $i_0 \ldots, i_{c - 1}, j_0, \ldots, j_{c - 1} \in [n]$ and all $l \in [c]$ with $i_k \neq j_k$ for all $k \in [c]$.
    Then, for every convex function $F \colon \R \to \R$, one has
    \begin{align}\label{eq:decoupling-gen}
    \begin{split}           
        &\ep{F\left(\sum_{\substack{i_0 \ldots, i_{c - 1}, j_0, \ldots, j_{c - 1} \in [n]\\ \forall k \in [c] : i_k \neq j_k}} f_{i_0, \ldots, i_{c - 1}, j_0, \ldots, j_{c - 1}}((X^{(0)}_{i_0}, \ldots, X^{(c - 1)}_{i_{c - 1}}), (X^{(0)}_{j_0}, \ldots, X^{(c - 1)}_{j_{c - 1}}))\right)}
            \\&\le \ep{F\left(4^c \sum_{i_0 \ldots, i_{c - 1}, j_0, \ldots, j_{c - 1} \in [n]} f_{i_0, \ldots, i_{c - 1}, j_0, \ldots, j_{c - 1}}((X^{(0)}_{i_0}, \ldots, X^{(c - 1)}_{i_{c - 1}}), (\hat{X}^{(0)}_{j_0}, \ldots, \hat{X}^{(c - 1)}_{j_{c - 1}}))\right)} \; ,
    \end{split}
    \end{align}
    where $(\hat{X}^{(k)}_i)_{i \in [n], k \in [c]}$ is an independent copy of $(X^{(k)}_i)_{i \in [n], k \in [c]}$.
\end{lemma}
\begin{proof}
    The will by induction on $c$, and the induction start, $c = 1$, is exactly \Cref{lem:decoupling-1d}.
    Now assume that $c > 1$ and that \cref{eq:decoupling-gen} holds for $c - 1$.
    
    We start by defining the functions, $g_{i_{c - 1}, j_{c - 1}} \colon T \times T \to \R$ for all $i_{c - 1}, j_{c - 1} \in [n]$, by
    \begin{align*}
        g_{i_{c - 1}, j_{c - 1}}(x, y) = \sum_{\substack{i_0 \ldots, i_{c - 2}, j_0, \ldots, j_{c - 2} \in [n]\\ \forall k \in [c - 1] : i_k \neq j_k}} f_{i_0, \ldots, i_{c - 1}, j_0, \ldots, j_{c - 1}}((X^{(0)}_{i_0}, \ldots, X^{(c - 2)}_{i_{c - 2}} x), (X^{(0)}_{j_0}, \ldots, X^{(c - 2)}_{j_{c - 2}}, y)) .
    \end{align*}
    We note that \cref{eq:decoupling-gen-cond} implies that
    \begin{align*}
        \epcond{g_{i_{c - 1}, j_{c - 1}}(X^{(c - 1)}_{i_{c - 1}}, X^{(c - 1)}_{j_{c - 1}})}{X^{(c - 1)}_{i_{c - 1}}, (X^{(k)}_{i})_{i \in [n], k \in [c - 1]}}
            &= 0 \; , \\
        \epcond{g_{i_{c - 1}, j_{c - 1}}(X^{(c - 1)}_{i_{c - 1}}, X^{(c - 1)}_{j_{c - 1}})}{X^{(c - 1)}_{j_{c - 1}}, (X^{(k)}_{i})_{i \in [n], k \in [c - 1]}}
            &= 0 \; ,
    \end{align*}
    for all $i_{c - 1} \neq j_{c - 1} \in [n]$.
    This implies that we can use \Cref{lem:decoupling-1d} while conditioning on $(X^{(k)}_{i})_{i \in [n], k \in [c - 1]}$ to get that
    \begin{align*}
        &\ep{F\left(\sum_{\substack{i_0 \ldots, i_{c - 1}, j_0, \ldots, j_{c - 1} \in [n]\\ \forall k \in [c] : i_k \neq j_k}} f_{i_0, \ldots, i_{c - 1}, j_0, \ldots, j_{c - 1}}((X^{(0)}_{i_0}, \ldots, X^{(c - 1)}_{i_{c - 1}}), (X^{(0)}_{j_0}, \ldots, X^{(c - 1)}_{j_{c - 1}}))\right)}
            \\&= \ep{\epcond{F\left(\sum_{i_{c - 1} \neq j_{c - 1} \in [n]} g_{i_{c - 1}, j_{c - 1}}(X^{(c - 1)}_{i_{c - 1}}, X^{(c - 1)}_{c - 1})\right)}{(X^{(k)}_{i})_{i \in [n], k \in [c - 1]}}}
            \\&\le \ep{\epcond{F\left(4 \sum_{i_{c - 1}, j_{c - 1} \in [n]} g_{i_{c - 1}, j_{c - 1}}(X^{(c - 1)}_{i_{c - 1}}, \hat{X}^{(c - 1)}_{c - 1})\right)}{(X^{(k)}_{i})_{i \in [n], k \in [c - 1]}}}
            \\&= \ep{F\left(4 \sum_{\substack{i_0 \ldots, i_{c - 1}, j_0, \ldots, j_{c - 1} \in [n]\\ \forall k \in [c - 1] : i_k \neq j_k}} f_{i_0, \ldots, i_{c - 1}, j_0, \ldots, j_{c - 1}}((X^{(0)}_{i_0}, \ldots, X^{(c - 1)}_{i_{c - 1}}), (X^{(0)}_{j_0}, \ldots, X^{(c - 2)}_{j_{c - 2}}, \hat{X}^{(c - 1)}_{j_{c - 1}}))\right)} .
    \end{align*}

    We then define the functions, $h_{i_0, \ldots, i_{c - 2}, j_0, \ldots, j_{c - 2}} \colon T^{c -1} \times T^{c -1} \to \R$ for all $i_0 \ldots, i_{c - 2}, j_0, \ldots, j_{c - 2} \in [n]$, by
    \begin{align*}
        &h_{i_0, \ldots, i_{c - 2}, j_0, \ldots, j_{c - 2}}((x_0, \ldots, x_{c - 2}), (y_0, \ldots, x_{c - 2}))
            \\&\qquad\qquad= \sum_{i_{c - 1}, j_{c - 1} \in [n]} f_{i_0, \ldots, i_{c - 1}, j_0, \ldots, j_{c - 1}}((x_{0}, \ldots, x_{c - 2}, X^{(c - 1)}_{i_{c - 1}}), (y_0, \ldots, y_{c - 2}, \hat{X}^{(c -1)}_{j_{c - 1}})) \; .
    \end{align*}
    Now let $i_0 \ldots, i_{c - 2}, j_0, \ldots, j_{c - 2} \in [n]$ and $l \in [c - 1]$ with $i_k \neq j_k$ for all $k \in [c - 1]$ and let $H = h_{i_0, \ldots, i_{c - 2}, j_0, \ldots, j_{c - 2}}$.
    Then again by \cref{eq:decoupling-gen-cond} we get that for 
    \begin{align*}
        \epcond{H((X^{(0)}_{i_0}, \ldots, X^{(c - 2)}_{i_{c - 2}}), (X^{(0)}_{j_0}, \ldots, X^{(c - 2)}_{j_{c - 2}}))}{(X^{(k)}_{i_k})_{k \in [c - 1] \setminus \set{l}}, (X^{(k)}_{j_k})_{k \in [c - 1]}, (X^{(c - 1)}_i, \hat{X}^{(c - 1)}_i)_{i \in [n]}} 
        &= 0 , \\
        \epcond{H((X^{(0)}_{i_0}, \ldots, X^{(c - 2)}_{i_{c - 2}}), (X^{(0)}_{j_0}, \ldots, X^{(c - 2)}_{j_{c - 2}}))}{(X^{(k)}_{i_k})_{k \in [c - 1]}, (X^{(k)}_{j_k})_{k \in [c - 1] \setminus \set{l}}, (X^{(c - 1)}_i, \hat{X}^{(c - 1)}_i)_{i \in [n]}} 
        &= 0  .
    \end{align*}
    We can then use the induction hypothesis conditioned on $(X^{(c - 1)}_i, \hat{X}^{(c - 1)}_i)_{i \in [n]}$ to get that\footnote{In the following calculations, we use $\mathbf{i}$ and $\mathbf{j}$ as shortcuts for $i_0, \ldots, i_{c - 2}$ and $j_0, \ldots, j_{c - 2}$ respectively.}
    \begin{align*}
        &\ep{F\left(4 \sum_{\substack{i_0 \ldots, i_{c - 1}, j_0, \ldots, j_{c - 1} \in [n]\\ \forall k \in [c - 1] : i_k \neq j_k}} f_{i_0, \ldots, i_{c - 1}, j_0, \ldots, j_{c - 1}}((X^{(0)}_{i_0}, \ldots, X^{(c - 1)}_{i_{c - 1}}), (X^{(0)}_{j_0}, \ldots, X^{(c - 2)}_{j_{c - 2}}, \hat{X}^{(c - 1)}_{j_{c - 1}}))\right)} 
            \\&= \ep{\epcond{F\left(4\sum_{\substack{i_0 \ldots, i_{c - 2}, j_0, \ldots, j_{c - 2} \in [n]\\ \forall k \in [c - 1] : i_k \neq j_k}} h_{\mathbf{i}, \mathbf{j}}((X^{(0)}_{i_0}, \ldots, X^{(c - 2)}_{i_{c - 2}}), (X^{(0)}_{j_0}, \ldots, X^{(c - 2)}_{j_{c - 2}}))\right)}{(X^{(c - 1)}_i, \hat{X}^{(c - 1)}_i)_{i \in [n]}}}
            \\&\le \ep{\epcond{F\left(4^{c}\sum_{i_0 \ldots, i_{c - 2}, j_0, \ldots, j_{c - 2} \in [n]} h_{\mathbf{i}, \mathbf{j}}((X^{(0)}_{i_0}, \ldots, X^{(c - 2)}_{i_{c - 2}}), (\hat{X}^{(0)}_{j_0}, \ldots, \hat{X}^{(c - 2)}_{j_{c - 2}}))\right)}{(X^{(c - 1)}_i, \hat{X}^{(c - 1)}_i)_{i \in [n]}}}
            \\&= \ep{F\left(4^{c}\sum_{i_0 \ldots, i_{c - 1}, j_0, \ldots, j_{c - 1} \in [n]} f_{i_0, \ldots, i_{c - 1}, j_0, \ldots, j_{c - 1}}((X^{(0)}_{i_0}, \ldots, X^{(c - 1)}_{i_{c - 1}}), (\hat{X}^{(0)}_{j_0}, \ldots, \hat{X}^{(c - 1)}_{j_{c - 1}}))\right)} .
    \end{align*}
    This finishes the induction step and thus the proof.
\end{proof}

We are now ready to prove a decoupling lemma for Mixed Tabulation.

\begin{lemma}\label{lem:mixed-decoupling}
    \begin{align}\label{eq:mixed-decouple-main}
        &\pnorm{\sum_{i \in [s]} \sum_{x, y \in [u]} \indicator{h_2(x) \neq h_2(y)} \sigma(i, x) \sigma(i, y) \indicator{h(i, x) = h(i, y)} w_x w_y}{p}
            \\&\le \sum_{\pi \in P_{partial}} \sum_{\tau \in P_{prefix}} 4^{c + 1} \pnorm{\sum_{i \in [s]} \sum_{x, y \in U_{\pi, \tau}}
                \sigma_{I_\pi, \abs{\tau}}(i, x) \overline{\sigma_{I_\pi, \abs{\tau}}}(i, y) \indicator{h_{I_\pi, \abs{\tau}}(i, x) = \overline{h_{I_\pi, \abs{\tau}}}(i, y)} w_x w_y
            }{p}
    \end{align}
    where 
    \begin{align*}
        h_{I, r}(i, x)
            &= \left(h_{1, I}(x) \xor h_{3, \set{r}}(h^{\set{r}}_{2, I}(x) \xor i \xor h^{\set{r}}_{2, I^c}(x)) \right) \xor \left(h_{1, I^c}(x) \xor h_{3, \set{r}^c}( h^{\set{r}^c}_2(x) \xor i^{\otimes (d - 1)}) \right) \\
        \sigma_{I, r}(i, x)
                &= \left(\sigma_{1, I}(x) \xor \sigma_{3, \set{r}}(h^{\set{r}}_{2, I}(x) \xor i \xor h^{\set{r}}_{2, I^c}(x)) \right) \xor \left(\sigma_{1, I^c}(x) \xor \sigma_{3, \set{r}^c}( h^{\set{r}^c}_2(x) \xor i^{\otimes (d - 1)}) \right) \\
        \overline{h_{I, r}}(i, y)
            &= \left(\overline{h_{1, I}}(y) \xor \overline{h_{3, \set{r}}}\big(\overline{h^{\set{r}}_{2, I}}(y) \xor i \xor h^{\set{r}}_{2, I^c}(y))\big) \right) \xor \left(h_{1, I^c}(y) \xor h_{3, \set{r}^c}( h^{\set{r}^c}_2(y) \xor i^{\otimes (d - 1)}) \right) \\
        \overline{\sigma_{I, r}}(i, y)
            &= \left(\overline{\sigma_{1, I}}(y) \xor \overline{\sigma_{3, \set{r}}}\big(\overline{h^{\set{r}}_{2, I}}(y) \xor i \xor h^{\set{r}}_{2, I^c}(y))\big) \right) \xor \left(\sigma_{1, I^c}(y) \xor \sigma_{3, \set{r}^c}( h^{\set{r}^c}_2(y) \xor i^{\otimes (d - 1)}) \right)
    \end{align*}

    Futhermore, we have that
    \begin{align}\label{eq:mixed-decouple-multiplicity}
        \sum_{\pi} \sum_{\tau} \sum_{x \in U_{\pi, \tau}} w_x^2 \le d 2^c \norm{w}{2}^2
    \end{align}
    for all $w \in \R^u$.
\end{lemma}
\begin{proof}
    We will start by proving \cref{eq:mixed-decouple-multiplicity}.
    Let $x \in [u]$, now it is easy to see that there exists $2^c$, $\pi \in P_{partial}$ such that $\pi \subseteq x$ and similarly there exists $d$, $\tau \in P_{prefix}$ such that $\tau \subseteq h_2(x)$.
    We there have that there $d 2^c$ pairs $(\pi, \tau) \in P_{partial} \times P_{prefix}$ such that $x \in U_{\pi, \tau}$ and thus \cref{eq:mixed-decouple-multiplicity} follows.

    We will now focus on proving \cref{eq:mixed-decouple-main}.
    We start by defining the events $A_{x, y, \tau}$ for $x, y \in [u]$ and $\tau \in P_{prefix}$ by $\indicator{A_{x, y, \tau}} = \indicator{h_{2, [\abs{\tau}]}(x) = h_{2, [\abs{\tau}]}(y) = \tau}\!\indicator{h_{2, \set{\abs{\tau}}}(x) \neq h_{2, \set{\abs{\tau}}}(y)}$.
    We can then write
    \begin{align*}
        &\pnorm{\sum_{i \in [s]} \sum_{x, y \in [u]} \indicator{h_2(x) \neq h_2(y)} \sigma(i, x) \sigma(i, y) \indicator{h(i, x) = h(i, y)} w_x w_y}{p}
            \\&= \pnorm{\sum_{\tau \in P_{prefix}} \sum_{i \in [s]} \sum_{x, y \in [u]} \indicator{A_{x, y, \tau}} \sigma(i, x) \sigma(i, y) \indicator{h(i, x) = h(i, y)} w_x w_y}{p}
            \\&\le \sum_{\tau \in P_{prefix}} \pnorm{\sum_{i \in [s]} \sum_{x, y \in [u]} \indicator{A_{x, y, \tau}} \sigma(i, x) \sigma(i, y) \indicator{h(i, x) = h(i, y)} w_x w_y}{p}
    \end{align*}
    We then define the functions $h_{r}, \overline{h_{r}} \colon [s] \times [u] \to [m/s]$ and $\sigma_{r}, \overline{\sigma_{r}} \colon [s] \times [u]$ by
    \begin{align*}
        h_{r}(i, x)
            &= h_1(x) \xor h_{3, \set{r}}(h_2^{\set{r}}(x) \xor i) \xor h_{3, \set{r}^c}(h^{\set{r}^c}_2(x) \xor i^{\otimes (d - 1)}) \\
        \sigma_{r}(i, x)
                &= \sigma_{1}(x) \sigma_{3, \set{r}}(h_2^{\set{r}}(x) \xor i) \cdot \sigma_{3, \set{r}^c}( h^{\set{r}^c}_2(x) \xor i^{\otimes (d - 1)}) \\
        \overline{h_{r}}(i, y)
            &= h_1(y) \xor \overline{h_{3, \set{r}}}(h_2^{\set{r}}(y) \xor i) \xor h_{3, \set{r}^c}(h^{\set{r}^c}_2(y) \xor i^{\otimes (d - 1)}) \\
        \overline{\sigma_{r}}(i, y)
            &= \sigma_{1}(x) \overline{\sigma_{3, \set{r}}}(h_2^{\set{r}}(x) \xor i) \cdot \sigma_{3, \set{r}^c}( h^{\set{r}^c}_2(x) \xor i^{\otimes (d - 1)})
    \end{align*}
    For a fixed $\tau \in P_{prefix}$, we see that the expression $$\sum_{i \in [s]} \sum_{x, y \in [u]} \indicator{A_{x, y, \tau}} \sigma(i, x) \sigma(i, y) \indicator{h(i, x) = h(i, y)} w_x w_y$$ satisfies the requirement of \Cref{lem:decoupling-1d} for the random variables $(h_{3, \set{\abs{\tau}}}(\alpha), \sigma_{3, \set{\abs{\tau}}}(\alpha))_{\alpha \in \Sigma}$.
    So applying \Cref{lem:decoupling-1d} we get that
    \begin{align*}
        &\sum_{\tau \in P_{prefix}} \pnorm{\sum_{i \in [s]} \sum_{x, y \in [u]} \indicator{A_{x, y, \tau}} \sigma(i, x) \sigma(i, y) \indicator{h(i, x) = h(i, y)} w_x w_y}{p}
            \\&\le \sum_{\tau \in P_{prefix}} 4 \pnorm{\sum_{i \in [s]} \sum_{x, y \in [u]} \indicator{x \neq y} \sigma_{\abs{\tau}}(i, x) \overline{\sigma_{\abs{\tau}}}(i, y) \indicator{h_{\abs{\tau}}(i, x) = \overline{h_{\abs{\tau}}}(i, y)} w_x w_y}{p}
    \end{align*}
    We can now rewrite this expression as follows,
    \begin{align*}
        &\sum_{\tau \in P_{prefix}} 4 \pnorm{\sum_{i \in [s]} \sum_{x, y \in [u]} \indicator{x \neq y} \sigma_{\abs{\tau}}(i, x) \overline{\sigma_{\abs{\tau}}}(i, y) \indicator{h_{\abs{\tau}}(i, x) = \overline{h_{\abs{\tau}}}(i, y)} w_x w_y}{p}
            \\&= \sum_{\tau \in P_{prefix}} 4 \pnorm{\sum_{\pi \in P_{partial}} \sum_{i \in [s]} \sum_{x, y \in [u]} \indicator{x \cap y = \pi} \sigma_{\abs{\tau}}(i, x) \overline{\sigma_{\abs{\tau}}}(i, y) \indicator{h_{\abs{\tau}}(i, x) = \overline{h_{\abs{\tau}}}(i, y)} w_x w_y}{p}
            \\&= \sum_{\pi \in P_{partial}} \sum_{\tau \in P_{prefix}} 4 \pnorm{\sum_{i \in [s]} \sum_{x, y \in [u]} \indicator{x \cap y = \pi} \sigma_{\abs{\tau}}(i, x) \overline{\sigma_{\abs{\tau}}}(i, y) \indicator{h_{\abs{\tau}}(i, x) = \overline{h_{\abs{\tau}}}(i, y)} w_x w_y}{p}
    \end{align*}
    For fixed $\pi \in P_{partial}$ and $\tau \in P_{prefix}$, we see that the expression $$\sum_{i \in [s]} \sum_{x, y \in [u]} \indicator{x \cap y = \pi} \sigma_{\abs{\tau}}(i, x) \overline{\sigma_{\abs{\tau}}}(i, y) \indicator{h_{\abs{\tau}}(i, x) = \overline{h_{\abs{\tau}}}(i, y)} w_x w_y$$ satisfies the requirement of \Cref{lem:decoupling-gen} for the random variables $(h_{1, I_\pi^c}(x), h_{2, I_{\pi}^c}(x), \sigma_{1, I_{\pi}^c}(x), \sigma_{2, I_{\pi}^c}(x) )_{x \in [u]}$.
    So applying \Cref{lem:decoupling-1d} we get that
    \begin{align*}
        &\sum_{\pi \in P_{partial}} \sum_{\tau \in P_{prefix}} 4 \pnorm{\sum_{i \in [s]} \sum_{x, y \in [u]} \indicator{x \cap y = \pi} \sigma_{\abs{\tau}}(i, x) \overline{\sigma_{\abs{\tau}}}(i, y) \indicator{h_{\abs{\tau}}(i, x) = \overline{h_{\abs{\tau}}}(i, y)} w_x w_y}{p}
            \\&\le \sum_{\pi \in P_{partial}} \sum_{\tau \in P_{prefix}} 4^{c + 1} \pnorm{\sum_{i \in [s]} \sum_{x, y \in [u]} \indicator{x \cap y = \pi} \sigma_{I_{\pi}, \abs{\tau}}(i, x) \overline{\sigma_{I_{\pi}, \abs{\tau}}}(i, y) \indicator{h_{I_{\pi}, \abs{\tau}}(i, x) = \overline{h_{I_{\pi}, \abs{\tau}}}(i, y)} w_x w_y}{p}
    \end{align*}
    This finishes the proof of \cref{eq:mixed-decouple-main} and thus the lemma.
\end{proof}

Finally, we can prove that Mixed Tabulation has an $(p, 4^{c + 2}, 4 \eps^{-2} 3^c \frac{s}{m} \abs{\Sigma}^{-d})$-decoupling-decomposition.
\begin{lemma}
    \begin{align*}
        &\prb{\abs{Z} \ge \eps}
        \\&\le \left( \eps^{-1} \sum_{\pi \in P_{partial}} \sum_{\tau \in P_{prefix}} \frac{4^{c + 4}}{s} \pnorm{\sum_{i \in [s]} \sum_{x, y \in U_{\pi, \tau}}
        \sigma_{I_\pi, \abs{\tau}}(i, x) \overline{\sigma_{I_\pi, \abs{\tau}}}(i, y) \indicator{h_{I_\pi, \abs{\tau}}(i, x) = \overline{h_{I_\pi, \abs{\tau}}}(i, y)} w_x w_y
    }{p} \right)^p 
    \\&\qquad\qquad+ 4 \eps^{-2} 3^c \frac{s}{m} \abs{\Sigma}^{-d} ,
    \end{align*}
    where 
    \begin{align*}
        h_{I, r}(i, x)
            &= \left(h_{1, I}(x) \xor h_{3, \set{r}}(h^{\set{r}}_{2, I}(x) \xor i \xor h^{\set{r}}_{2, I^c}(x)) \right) \xor \left(h_{1, I^c}(x) \xor h_{3, \set{r}^c}( h^{\set{r}^c}_2(x) \xor i^{\otimes (d - 1)}) \right) \\
        \sigma_{I, r}(i, x)
                &= \left(\sigma_{1, I}(x) \xor \sigma_{3, \set{r}}(h^{\set{r}}_{2, I}(x) \xor i \xor h^{\set{r}}_{2, I^c}(x)) \right) \xor \left(\sigma_{1, I^c}(x) \xor \sigma_{3, \set{r}^c}( h^{\set{r}^c}_2(x) \xor i^{\otimes (d - 1)}) \right) \\
        \overline{h_{I, r}}(i, y)
            &= \left(\overline{h_{1, I}}(y) \xor \overline{h_{3, \set{r}}}\big(\overline{h^{\set{r}}_{2, I}}(y) \xor i \xor h^{\set{r}}_{2, I^c}(y))\big) \right) \xor \left(h_{1, I^c}(y) \xor h_{3, \set{r}^c}( h^{\set{r}^c}_2(y) \xor i^{\otimes (d - 1)}) \right) \\
        \overline{\sigma_{I, r}}(i, y)
            &= \left(\overline{\sigma_{1, I}}(y) \xor \overline{\sigma_{3, \set{r}}}\big(\overline{h^{\set{r}}_{2, I}}(y) \xor i \xor h^{\set{r}}_{2, I^c}(y))\big) \right) \xor \left(\sigma_{1, I^c}(y) \xor \sigma_{3, \set{r}^c}( h^{\set{r}^c}_2(y) \xor i^{\otimes (d - 1)}) \right)
    \end{align*}

    Futhermore, we have that
    \begin{align*}
        \sum_{\pi} \sum_{\tau} \sum_{x \in U_{\pi, \tau}} w_x^2 \le d 2^c \norm{w}{2}^2
    \end{align*}
    for all $w \in \R^u$.
\end{lemma}
\begin{proof}
    We use a union bound to obtain that,
    \begin{align*}
        \prb{\abs{Z} \ge \eps}
            &\le \prb{\abs{\sum_{i \in [s]} \sum_{x\neq y \in [u]} \sigma_1(x)\sigma_1(y) \indicator{(h_1(x), h_2(x)) = (h_1(y), h_2(y))} w_x w_y  } \ge \eps/2 }
            \\&+ \prb{\abs{\sum_{i \in [s]} \sum_{x, y \in [u]} \indicator{h_2(x) \neq h_2(y)} \sigma(i, x) \sigma(i, y) \indicator{h(i, x) = h(i, y)} w_x w_y} \ge \eps/2}
    \end{align*}
    For the first term, we use that $\ep{(\sum_{i \in [s]} \sum_{x\neq y \in [u]} \sigma_1(x)\sigma_1(y) \indicator{(h_1(x), h_2(x)) = (h_1(y), h_2(y)) w_x w_y}  )^2} \le 3^c \norm{w}{2}^4 \frac{s}{m} \abs{\Sigma}^{-d}$.
    So applying Markov's inequality, we get that,
    \begin{align*}
        \prb{\abs{\sum_{i \in [s]} \sum_{x\neq y \in [u]} \sigma_1(x)\sigma_1(y) \indicator{(h_1(x), h_2(x)) = (h_1(y), h_2(y))} w_x w_y  } \ge \eps/2 }
            \le 4 \eps^{-2} 3^c \frac{s}{m} \abs{\Sigma}^{-d} .
    \end{align*}
    For the second term we apply Markov's inequality for $p$ and then the result follows by \Cref{lem:mixed-decoupling}.
\end{proof}

\subsection{Concentration}\label{sec:mixed-concentration}
The goal of this section is to show that Mixed Tabulation is $(p, \gamma_p^c)$-strongly-concentrated where $\gamma_p = K c \max\!\set{1, \frac{p}{\log \abs{\Sigma}}}$ for a universal constant $K$.
This is done in next 3 lemmas that prove the individual parts of the strong concentration.
They all follow the same blueprint as the results in~\cite{Houen22}.

\begin{lemma}\label{lem:mixed-property-1}
    For any value function, $v \colon [s] \times [m/s]$, and any vector $w \in \R^U$ the following concentration results for Mixed Tabulation hashing holds
    \begin{align}\label{eq:mixed-concentration-poisson}
        \pnorm{\sum_{i \in [s]} \sum_{x \in U} \sigma(i, x) v(i, h(i, x)) w_x}{p}
            \le \Psi_p\left(\gamma_p^c \norm{v}{\infty} \norm{w}{\infty}, \gamma_p^c \frac{s}{m} norm{v}{2}^2 \norm{w}{2}^2 \right)
        ,
    \end{align}
    \begin{align}\label{eq:mixed-concentration-gaussian}
        \pnorm{\sum_{i \in [s]} \sum_{x \in U} \sigma(i, x) v(i, h(i, x)) w_x}{p}
            \le \sqrt{\gamma_p^c \frac{p}{\log m/s}} \norm{v}{2} \norm{w}{2}
        .
    \end{align}
    Here $\gamma_p = K c \max\!\set{1, \frac{p}{\log \abs{\Sigma}}}$ for a universal constant $K$.
\end{lemma}
\begin{proof}
    We start by rewriting the expression
    \begin{align*}
        \pnorm{\sum_{i \in [s]} \sum_{x \in U} \sigma(i, x) v(i, h(i, x)) w_x}{p}
            = \pnorm{\sum_{\alpha \in \Sigma} \sigma_2(\alpha) \sum_{i \in [s]} \sum_{x \in U} \sigma_1(x) v(i, h_1(x) \xor h_3(x)) \indicator{h_2(x) = \alpha \xor i} w_x }{p}
    \end{align*}
    We define the value function $v' \colon U \times (\Sigma \times [m/s]) \to \R$ by $v'(x, (i, j)) = v(i, j)\indicator{i \in [s]} w_x$.
    We can then write our expression as
    \begin{align*}
        &\pnorm{\sum_{\alpha \in \Sigma} \sigma_2(\alpha) \sum_{i \in [s]} \sum_{x \in U} \sigma_1(x) v(i, h_1(x) \xor h_3(x)) \indicator{h_2(x) = \alpha \xor i} w_x }{p}
            \\&= \pnorm{\sum_{\alpha \in \Sigma} \sigma_2(\alpha) \sum_{x \in U} v'(x, (\alpha \xor h_2(x), h_3(\alpha) \xor h_1(x))) }{p}
    \end{align*}

    Now we start by proving \cref{eq:mixed-concentration-gaussian}.
    We use \Cref{lem:fully-random-gauss} to get that
    \begin{align*}
        &\pnorm{\sum_{\alpha \in \Sigma} \sigma_2(\alpha) \sum_{x \in U} v'(x, (\alpha \xor h_2(x), h_3(\alpha) \xor h_1(x))) }{p}
            \\&\le e \sqrt{\frac{p}{\log m/s}}\pnorm{\sum_{\alpha \in \Sigma} \sum_{j \in [m/s]} \left(\sum_{x \in U}  v'(x, (\alpha \xor h_2(x), j \xor h_1(x))) \right)^2 }{p/2}^{1/2}
    \end{align*}
    Now we use \Cref{lem:simple-squares} to obtain that
    \begin{align}\label{eq:sum-squares}
        \begin{split}            
        \pnorm{\sum_{\alpha \in \Sigma} \sum_{j \in [m/s]} \left(\sum_{x \in U}  v'(x, (\alpha \xor h_2(x), j \xor h_1(x))) \right)^2 }{p/2}
            &\le \gamma_p^c \sum_{x \in U} \sum_{i \in \Sigma} \sum_{j \in [m/s]} v'(x, (i, j))^2
            \\&= \gamma_p^c \sum_{x \in U} \sum_{i \in \Sigma} \sum_{j \in [m/s]} v(i, j)^2\indicator{i \in [s]} w_x^2
            \\&= \gamma_p^c \norm{v}{2}^2 \norm{w}{2}^2
        \end{split}
    \end{align}
    This then give us that
    \begin{align*}
        &\pnorm{\sum_{\alpha \in \Sigma} \sigma_2(\alpha) \sum_{x \in U} v'(x, (\alpha \xor h_2(x), h_3(\alpha) \xor h_1(x))) }{p}
            \\&\le e \sqrt{\gamma_p^c \frac{p}{\log m/s}} \norm{v}{2} \norm{w}{2}
    \end{align*}
    This finishes the proof of \cref{eq:mixed-concentration-gaussian}.

    Now we focus on \cref{eq:mixed-concentration-poisson}.
    We use \Cref{lem:fully-random-concentration} to get that
    \begin{align*}
        &\pnorm{\sum_{\alpha \in \Sigma} \sigma_2(\alpha) \sum_{x \in U} v'(x, (\alpha \xor h_2(x), h_3(\alpha) \xor h_1(x))) }{p}
            \\&\le L \pnorm{\Psi_p\left( A, B \right)}{p}
    \end{align*}
    where
    \begin{align*}
        A &= \max_{\alpha \in \Sigma} \max_{j \in [m/s]} \abs{\sum_{x \in U} v'(x, (\alpha \xor h_2(x), h_3(\alpha) \xor h_1(x)))} \\
        B &= \sum_{\alpha \in \Sigma} \sum_{j \in [m/s]} \left(\sum_{x \in U}  v'(x, (\alpha \xor h_2(x), j \xor h_1(x))) \right)^2
    \end{align*}
    We apply \Cref{lem:fn-moment} to get that we just need to bound the two expressions
    \begin{align*}
        \pnorm{\max_{\alpha \in \Sigma} \max_{j \in [m/s]} \abs{\sum_{x \in U} v'(x, (\alpha \xor h_2(x), h_3(\alpha) \xor h_1(x)))}}{p} \\
        \pnorm{\sum_{\alpha \in \Sigma} \sum_{j \in [m/s]} \left(\sum_{x \in U}  v'(x, (\alpha \xor h_2(x), j \xor h_1(x))) \right)^2}{p/2}
    \end{align*}
    From \cref{eq:sum-squares} we have that
    \begin{align*}
        \pnorm{\sum_{\alpha \in \Sigma} \sum_{j \in [m/s]} \left(\sum_{x \in U}  v'(x, (\alpha \xor h_2(x), j \xor h_1(x))) \right)^2}{p/2}
            \le \gamma_p^c \norm{v}{2}^2 \norm{w}{2}^2
    \end{align*}
    Let $\bar{p} = \max\!\set{p, \log(m/s \cdot \abs{\Sigma}}$ and use \Cref{lem:simple-concentration} to get that
    \begin{align*}
        &\pnorm{\max_{\alpha \in \Sigma} \max_{j \in [m/s]} \abs{\sum_{x \in U} v'(x, (\alpha \xor h_2(x), j \xor h_1(x)))}}{p}
            \\&\le \pnorm{\max_{\alpha \in \Sigma} \max_{j \in [m/s]} \abs{\sum_{x \in U} v'(x, (\alpha \xor h_2(x), j \xor h_1(x)))}}{\bar{p}}
            \\&\le \left( \sum_{\alpha \in \Sigma} \sum_{j \in [m/s]} \pnorm{\sum_{x \in U} v'(x, (\alpha \xor h_2(x), j \xor h_1(x)))}{\bar{p}}^{1/\bar{p}} \right)^{1/\bar{p}}
            \\&\le e \Psi_{\bar{p}}\left(\gamma_p^c \max_{x \in U} \max_{i \in [s]} \max_{j \in [m/s]} \abs{v(i, j) w_x}, \gamma_p^c \frac{s}{m \abs{\Sigma}} \sum_{x \in U} \sum_{i \in [s]} \sum_{j \in [m/s]} v(i, j)^2 w_x^2  \right)
            \\&\le e \Psi_{\bar{p}}\left(\gamma_p^c \norm{v}{\infty} \norm{w}{\infty}, \gamma_p^c \frac{1}{\abs{\Sigma}} \norm{v}{2}^2 \norm{w}{2}^2\right)
    \end{align*}
    We then get that
    \begin{align*}
        &\pnorm{\sum_{\alpha \in \Sigma} \sigma_2(\alpha) \sum_{x \in U} v'(x, (\alpha \xor h_2(x), h_3(\alpha) \xor h_1(x))) }{p}
            \\&\le e \Psi_p\left(e \Psi_{\bar{p}}\left(\gamma_p^c \norm{v}{\infty} \norm{w}{\infty}, \gamma_p^c \frac{1}{\abs{\Sigma}} \norm{v}{2}^2 \norm{w}{2}^2\right) , \gamma_p^c \frac{s}{m} \norm{v}{2}^2 \norm{w}{2}^2 \right)
    \end{align*}
    Since $s \le \sqrt{\abs{\Sigma}}$, the same case analysis as in the proof of \cite[Theorem 7]{Houen22} give us that since 
    \begin{align*}
        &\pnorm{\sum_{\alpha \in \Sigma} \sigma_2(\alpha) \sum_{x \in U} v'(x, (\alpha \xor h_2(x), h_3(\alpha) \xor h_1(x))) }{p}
            \le \Psi_p\left( \gamma_p^c \norm{v}{\infty} \norm{w}{\infty}, \gamma_p^c \frac{s}{m} \norm{v}{2}^2 \norm{w}{2}^2 \right)
    \end{align*}
    This finishes the proof.
\end{proof}

\begin{lemma}\label{lem:mixed-property-2}
    For any vector $w \in \R^U$, the following concentration result holds for Mixed Tabulation hashing,
    \begin{align*}
        \pnorm{\max_{i \in [s], j \in [m/s]}
            \abs{\sum_{x \in U} \sigma(i, x) \indicator{h(i, x) = j} w_x}
        }{p}
            \le \sqrt{\gamma_p^c  \frac{\log m}{\log m/s}} \norm{w}{2}
    \end{align*}
    Here $\gamma_p = K c \max\!\set{1, \frac{p}{\log \abs{\Sigma}}}$ for a universal constant $K$.
\end{lemma}
\begin{proof}
    We start by rewriting the expression,
    \begin{align*}
        &\pnorm{\max_{i \in [s], j \in [m/s]}
            \abs{\sum_{x \in U} \sigma(i, x) \indicator{h(i, x) = j} w_x}
        }{p}
            \\&= \pnorm{\max_{i \in [s], j \in [m/s]}
                \abs{\sum_{\alpha \in \Sigma} \sigma_3(\alpha) \sum_{x \in U} \sigma_1(x) \indicator{h_1(x) \xor h_3(\alpha) = j} \indicator{h_2(x) = \alpha \xor i} w_x}
            }{p} .
    \end{align*}
    Now we fix the randomness of $h_1, \sigma_1, h_2$ and use \Cref{lem:fully-random-gauss} to get that,
    \begin{align*}
        &\pnorm{\max_{i \in [s], j \in [m/s]}
                \abs{\sum_{\alpha \in \Sigma} \sigma_3(\alpha) \sum_{x \in U} \sigma_1(x) \indicator{h_1(x) \xor h_3(\alpha) = j} \indicator{h_2(x) = \alpha \xor i} w_x}
            }{p} 
            \\&\le \pnorm{\max_{i \in [s], j \in [m/s]}
                \abs{\sum_{\alpha \in \Sigma} \sigma_3(\alpha) \sum_{x \in U} \sigma_1(x) \indicator{h_1(x) \xor h_3(\alpha) = j} \indicator{h_2(x) = \alpha \xor i} w_x}
            }{\log m} .
            \\&\le e \max_{i \in [s], j \in [m/s]} \pnorm{
                \abs{\sum_{\alpha \in \Sigma} \sigma_3(\alpha) \sum_{x \in U} \sigma_1(x) \indicator{h_1(x) \xor h_3(\alpha) = j} \indicator{h_2(x) = \alpha \xor i} w_x}
            }{\log m}
            \\&\le e^2 \sqrt{\frac{\log m}{\log m/s}} \max_{i \in [s], j \in [m/s]} \sqrt{\sum_{\alpha \in \Sigma} \sum_{k \in [m/s]} \left( \sum_{x \in U} \sigma_1(x) \indicator{h_1(x) = j \xor k} \indicator{h_2(x) = \alpha \xor i} w_x \right)^2}
    \end{align*}
    We now note that the expression $\sqrt{\sum_{\alpha \in \Sigma} \sum_{k \in [m/s]} \left( \sum_{x \in U} \sigma_1(x) \indicator{h_1(x) = j \xor k} \indicator{h_2(x) = \alpha \xor i} w_x \right)^2}$ does not depend on $i$ and $j$, so might as well look at $i = j = 0$.
    We thus get that
    \begin{align*}
        &\pnorm{\max_{i \in [s], j \in [m/s]}
            \abs{\sum_{x \in U} \sigma(i, x) \indicator{h(i, x) = j} w_x}
        }{p}
            \\&\le e^2 \sqrt{\frac{\log m}{\log m/s}} \pnorm{\sum_{\alpha \in \Sigma} \sum_{k \in [m/s]} \left( \sum_{x \in U} \sigma_1(x) \indicator{h_1(x) = k} \indicator{h_2(x) = \alpha} w_x \right)^2}{p/2}^{1/2}
    \end{align*}
    We can then apply \Cref{lem:simple-squares} and obtain
    \begin{align*}
        &\pnorm{\max_{i \in [s], j \in [m/s]}
            \abs{\sum_{x \in U} \sigma(i, x) \indicator{h(i, x) = j} w_x}
        }{p}
            \\&\le \sqrt{\gamma_p^c  \frac{\log m}{\log m/s}} \norm{w}{2}
    \end{align*}
    This finishes the proof.
\end{proof}

\begin{lemma}\label{lem:mixed-property-3}
    For any vector $w \in \R^U$, the following concentration result holds for Mixed Tabulation hashing,
    \begin{align*}
        \pnorm{\sum_{i \in [s]} \sum_{j \in [m/s]} \left(\sum_{x \in U} \sigma(i, x) \indicator{h(i, x) = j} w_x \right)^2}{p}
            \le \gamma_p^c \max\!\set{s \norm{w}{2}^2, \frac{p}{\log m/s} \norm{w}{2}^2}
    \end{align*}
    Here $\gamma_p = K c \max\!\set{1, \frac{p}{\log \abs{\Sigma}}}$ for a universal constant $K$.
\end{lemma}
\begin{proof}
    We start by rewriting the expression
    \begin{align*}
        &\pnorm{\sum_{i \in [s]} \sum_{j \in [m/s]} \left(\sum_{x \in U} \sigma(i, x) \indicator{h(i, x) = j} w_x \right)^2}{p}
            \\&\qquad\qquad= \pnorm{\sum_{i \in [s]} \sum_{x, y \in U} \sigma(i, x) \sigma(i, y) \indicator{h(i, x) = h(i, y)} w_x w_y}{p}
            \\&\qquad\qquad\le \pnorm{\sum_{i \in [s]} \sum_{x, y \in U} \sigma(i, x) \sigma(i, y) \indicator{(h_1(x), h_2(x)) = (h_1(y), h_2(y))} w_x w_y}{p}
                \\&\qquad\qquad\qquad+ \pnorm{\sum_{i \in [s]} \sum_{x, y \in U} \sigma(i, x) \sigma(i, y) \indicator{h(i, x) = h(i, y)} \indicator{h_2(x) \neq h_2(y)} w_x w_y}{p}
            \\&\qquad\qquad= s \pnorm{\sum_{z \in \Sigma^d} \sum_{j \in [m/s]} \left( \sum_{x \in U} \sigma_1(x) \indicator{(h_1(x), h_2(x)) = (j, z)} w_x \right)^2}{p}
                \\&\qquad\qquad\qquad+ \pnorm{\sum_{i \in [s]} \sum_{x, y \in U} \sigma(i, x) \sigma(i, y) \indicator{h(i, x) = h(i, y)} \indicator{h_2(x) \neq h_2(y)} w_x w_y}{p}
    \end{align*}
    We will bound each of the two expressions separately.
    We use \Cref{lem:simple-squares} to get that
    \begin{align*}
        s \pnorm{\sum_{z \in \Sigma^d} \sum_{j \in [m/s]} \left( \sum_{x \in U} \sigma_1(x) \indicator{(h_1(x), h_2(x)) = (j, z)} w_x \right)^2}{p}
            \le \gamma_p^c s \norm{w}{2}^2
    \end{align*}
    To bound the other expression, we use \Cref{lem:mixed-decoupling} to get that,
    \begin{align*}
        &\pnorm{\sum_{i \in [s]} \sum_{x, y \in U} \sigma(i, x) \sigma(i, y) \indicator{h(i, x) = h(i, y)} \indicator{h_2(x) \neq h_2(y)} w_x w_y}{p}
            \\&\le \sum_{\pi \in P_{partial}} \sum_{\tau \in P_{prefix}} 4^{c + 1} \pnorm{\sum_{i \in [s]} \sum_{x, y \in U_{\pi, \tau}} \sigma_{I_\pi, \abs{\tau}}(i, x) \overline{\sigma_{I_\pi, \abs{\tau}}}(i, y) \indicator{h_{I_\pi, \abs{\tau}}(i, x) = \overline{h_{I_\pi, \abs{\tau}}}(i, y)} w_x w_y}{p}
    \end{align*}
    For each $\pi \in P_{partial}$ and each $\tau \in P_{prefix}$ this corresponds to a Mixed Tabulation hash function with $c' \le 2c$ and $d' \le 2$.
    We can then use \cref{eq:mixed-concentration-gaussian} to get that
    \begin{align*}
        &\sum_{\pi \in P_{partial}} \sum_{\tau \in P_{prefix}} 4^{c + 1} \pnorm{\sum_{i \in [s]} \sum_{x, y \in U_{\pi, \tau}} \sigma_{I_\pi, \abs{\tau}}(i, x) \overline{\sigma_{I_\pi, \abs{\tau}}}(i, y) \indicator{h_{I_\pi, \abs{\tau}}(i, x) = \overline{h_{I_\pi, \abs{\tau}}}(i, y)} w_x w_y}{p}
            \\&\le \sum_{\pi \in P_{partial}} \sum_{\tau \in P_{prefix}} 4^{c + 1} \gamma_p^c \sum_{x \in U_{\pi, \tau}} w_x^2
    \end{align*}
    We then use \cref{eq:mixed-decouple-multiplicity} to get that,
    \begin{align*}
        \sum_{\pi \in P_{partial}} \sum_{\tau \in P_{prefix}} 4^{c + 1} \gamma_p^c \sum_{x \in U_{\pi, \tau}} w_x^2
            \le \gamma_p^c d 8^c \norm{w}{2}^2
    \end{align*}
    Combing the bounds we get that
    \begin{align*}
        \pnorm{\sum_{i \in [s]} \sum_{j \in [m/s]} \left(\sum_{x \in U} \sigma(i, x) \indicator{h(i, x) = j} w_x \right)^2}{p}
            &\le \gamma_p^c + \gamma_p^c d 8^c \norm{w}{2}^2
            \\&\le \gamma_p^c \max\!\set{s \norm{w}{2}^2, \frac{p}{\log m/s} \norm{w}{2}^2}
    \end{align*}
    as wanted, which finishes the proof.
\end{proof}

\bibliographystyle{plain} 
\bibliography{paper}



\end{document}